\documentclass{jfp1}
\usepackage{proof}

\def\limplies{\supset}
\def\DML{\underline{\bf DML}}
\def\dmllang{\mbox{\underline{DML}}}
\def\ATS{\underline{\bf ATS}}
\def\PTS{\underline{\bf PTS}}
\def\atslang{\mbox{\underline{ATS}}}

\def\lamdyn{\lambda_{\it dyn}}

\def\ATSzero{\mbox{ATS}_0}
\def\ATSproof{\mbox{ATS}_{\it pf}}

\title
[Journal of Functional Programming]
{Applied Type System: An Approach to\break Practical Programming with Theorem-Proving}

\author
[Hongwei Xi]
{Hongwei Xi
\thanks{Supported in part by NSF grants no. CCR-0224244, no. CCR-0229480, and no. CCF-0702665}\\
Boston University, Boston, MA 02215, USA\\
\email{hwxi@cs.bu.edu}
}

\jdate{January, 2016}

\begin{document}

\label{firstpage}

\maketitle

\begin%
{abstract}
The framework Pure Type System ($\PTS$) offers a simple and general
approach to designing and formalizing type systems. However, in the
presence of dependent types, there often exist certain acute problems
that make it difficult for $\PTS$ to directly accommodate many common
realistic programming features such as general recursion, recursive
types, effects (e.g., exceptions, references, input/output), etc. In
this paper, Applied Type System ($\ATS$) is presented as a framework
for designing and formalizing type systems in support of practical
programming with advanced types (including dependent types). In
particular, it is demonstrated that $\ATS$ can readily accommodate a
paradigm referred to as programming with theorem-proving (PwTP) in
which programs and proofs are constructed in a syntactically
intertwined manner, yielding a practical approach to internalizing
constraint-solving needed during type-checking. The key salient
feature of $\ATS$ lies in a complete separation between statics, where
types are formed and reasoned about, and dynamics, where programs are
constructed and evaluated. With this separation, it is no longer
possible for a program to occur in a type as is otherwise allowed in
$\PTS$.  The paper contains not only a formal development of $\ATS$
but also some examples taken from $\atslang$, a programming language
with a type system rooted in $\ATS$, in support of employing $\ATS$ as a
framework to formulate advanced type systems for practical programming.
\end{abstract}

\tableofcontents

\section
{Introduction}
A primary motivation for developing Applied Type System ($\ATS$) stems
from an earlier attempt to support a restricted form of dependent types in
practical programming~\cite{DML-jfp07}.
While there is already a framework Pure Type System
($\PTS$)~\cite{PTS-lcwt92} that offers a simple and general approach to
designing and formalizing type systems, it is well understood that there
often exist some acute problems (in the presence of dependent types) making
it difficult for $\PTS$ to accommodate many common realistic programming
features. In particular, various studies reported in the literature
indicate that great efforts are often required in order to maintain a style
of pure reasoning as is advocated in $\PTS$ when features such as general
recursion~\cite{CONSTABLE87}, recursive types~\cite{MENDLER87},
effects~\cite{HMST95}, exceptions~\cite{HN88} and input/output are present.

The framework $\ATS$ is formulated to allow for designing and
formalizing type systems that can readily support common realistic
programming features. The formulation of $\ATS$ given in this paper is
primarily based on the work reported in two previous
papers~\cite{ATS-types03,CPwTP-icfp05} but there are some fundamental
changes in terms of the handling of proofs and proof construction.  In
particular, the requirement is dropped that a proof in $\ATS$ must be
represented as a normalizing lambda-term~\cite{ATSLF-08}.

In contrast to $\PTS$, the key salient feature of $\ATS$ lies in a
complete separation between statics, where types are formed and
reasoned about, from dynamics, where programs are constructed and
evaluated. This separation, with its origin in a previous study on a
restricted form of dependent types developed in Dependent ML
($\DML$)~\cite{DML-jfp07}, makes it straightforward to support
dependent types in the presence of effects such as references and
exceptions. Also, with the introduction of two new (and thus somewhat
unfamiliar) forms of types: {\em guarded types} and {\em asserting
  types}, $\ATS$ is able to capture program invariants in a manner
that is similar to the use of pre-conditions and
post-conditions~\cite{HOARE69}. By now, studies have shown amply and
convincingly that a variety of traditional programming paradigms
(e.g., functional programming, object-oriented programming,
meta-programming, modular programming) can be directly supported in
$\ATS$ without relying on {\em ad hoc} extensions, attesting to the
expressiveness of $\ATS$.  In this paper, the primary focus of study
is set on a novel programming paradigm referred to as {\em programming
  with theorem-proving} (PwTP) and its support in $\ATS$. In
particular, a type-theoretical foundation for PwTP is to be formally
established and its correctness proven.

The notion of type equality plays a pivotal r{\^o}le in type system
design. However, the importance of this r{\^o}le is often less evident
in commonly studied type systems. For instance, in the simply typed
$\lambda$-calculus, two types are considered equal if and only if they
are syntactically the same; in the second-order polymorphic
$\lambda$-calculus ($\lambda_2$)~\cite{REYNOLDS72A} and System
F~\cite{GIRARD86}, two types are considered equal if and only if they
are $\alpha$-equivalent; in the higher-order polymorphic
$\lambda$-calculus ($\lambda_\omega$), two types are considered equal
if and only if they are $\beta\eta$-equivalent.  This situation
immediately changes in $\ATS$, and let us see a simple example that
stresses this point.

\def\timp{\rightarrow}
\def\cnil{\mbox{\tt nil}}
\def\ccons{\mbox{\tt cons}}
\def\sint{\mbox{\it int}}
\def\snat{\mbox{\it nat}}
\def\sbool{\mbox{\it bool}}
\def\sprop{\mbox{\it prop}}
\def\stype{\mbox{\it type}}
\def\tint{\mbox{\bf int}}
\def\tlist{\mbox{\bf list}}
\begin{figure}[h]
\begin{verbatim}
fun
append
{a:type}{m,n:nat}
(
  xs: list (a, m), ys: list (a, n)
) : list (a, m+n) =
  case xs of
  | nil() => ys (* the first clause *)
  | cons(x, xs) => cons (x, append(xs, ys)) (* the second clause *)
// end of [append]
\end{verbatim}
\caption{List-append in $\atslang$}
\label{figure:list_append_function}
\end{figure}
In Figure~\ref{figure:list_append_function}, the presented code
implements a function in $\atslang$~\cite{ats-lang}, which is a
substantial system such that its compiler alone currently consists of
more than 165K lines of code implemented in $\atslang$
itself.\footnote{Please see {\tt http://www.ats-lang.org} for more details.}
The concrete syntax used in the implementation should be accessible to
those who are familiar with Standard ML (SML)~\cite{SML97}). Note that
$\atslang$ is a programming language equipped with a type system
rooted in $\ATS$, and the name of $\atslang$ derives from that of
$\ATS$. The type constructor $\tlist$ takes two arguments; when
applied to a type $T$ and an integer $I$, $\tlist(T,I)$ forms a type
for lists of length $I$ in which each element is of type $T$.  Also,
the two list constructors $\cnil$ and $\ccons$ are assigned the
following types:
\begin{center}
\[%
\begin%
{array}{ccl}
\cnil &~:~& \forall a:\stype.~()\timp\tlist (a, 0) \\
\ccons &~:~& \forall a:\stype.\forall n:\snat.~(a, \tlist (a, n))\timp\tlist(a, n+1) \\
\end{array}\]
\end{center}
So $\cnil$ constructs a list of length $0$, and $\ccons$ takes an element
and a list of length $n$ to form a list of length $n+1$.
\def\feval{\mbox{\tt eval}}
\def\fappend{\mbox{\tt append}}
The header of the function $\fappend$ indicates that $\fappend$ is assigned
the following type:
\begin%
{center}
\[%
\begin%
{array}{l}
\forall a:\stype.\forall m:\snat.\forall n:\snat.~(\tlist(a, m), \tlist(a, n))\timp\tlist(a, m+n) \\
\end{array}\]
\end{center}
which simply means that $\fappend$ returns a list of length $m+n$
when applied to one list of length $m$ and another list of length $n$.
Note that $\stype$ is a built-in sort in $\ATS$, and a static term of
the sort $\stype$ stands for a type (for dynamic terms). Also, $\sint$
is a built-in sort for integers in $\ATS$, and $\snat$ is the subset
sort $\{a:\sint\mid a\geq 0\}$ for all nonnegative integers.

When the above implementation of $\fappend$ is type-checked, the
following two constraints are generated:
\[\begin{array}{cl}
1. & \forall m:nat.\forall n:nat.~m=0 \limplies n=m+n \\
2. & \forall m:nat.\forall n:nat.\forall m':nat.~m=m'+1 \limplies (m'+n)+1 = m+n \\
\end{array}\]
The first constraint is generated when the first clause is
type-checked, which is needed for determining whether the types
$\tlist(a,n)$ and $\tlist(a,m+n)$ are equal under the assumption that
$\tlist(a,m)$ equals $\tlist(a,0)$. Similarly, the second constraint
is generated when the second clause is type-checked, which is needed
for determining whether the types $\tlist(a,(m'+n)+1)$ and
$\tlist(a,m+n)$ are equal under the assumption that $\tlist(a,m)$
equals $\tlist(a,m'+1)$. Clearly, certain restrictions need to be
imposed on the form of constraints allowed in practice so that an
effective approach can be found to perform constraint-solving.  In
$\dmllang$, a programming language based on $\DML$~\cite{DML-jfp07},
the constraints generated during type-checking are required to be
linear inequalities on integers so that the problem of constraint
satisfaction can be turned into the problem of linear integer
programming, for which there are many highly practical solvers (albeit
the problem of linear integer programming itself is NP-complete). This
is indeed a very simple design, but it can also be too restrictive,
sometimes, as nonlinear constraints (e.g., $\forall n:int. n*n\geq 0$)
are commonly encountered in practice. Furthermore, the very nature of
such a design indicates its being inherently {\em ad hoc}.

\def\xs{\mbox{\it xs}}
\def\ys{\mbox{\it ys}}
By combining programming with theorem-proving, a fundamentally
different design of constraint-solving can provide the programmer with
an option to handle nonlinear constraints through explicit proof
construction.  For the sake of a simpler presentation, let us assume
for this moment that even the addition function on integers cannot
appear in the constraints generated during type-checking.  Under such a
restriction, it is still possible to implement a list-append function
in $\atslang$ that is assigned a type capturing the invariant that the
length of the concatenation of two given lists $\xs$ and $\ys$ equals
$m+n$ if $\xs$ and $\ys$ are of length $m$ and $n$, respectively.  Let
us first see such an implementation given in
Figure~\ref{figure:mylist_append_1}, which is presented here as a
motivating example for programming with theorem-proving (PwTP).

\def\tZ{\mbox{\bf Z}}
\def\tS{\mbox{\bf S}}
\def\mynil{\mbox{\tt mynil}}
\def\mycons{\mbox{\tt mycons}}
\def\mylist{\mbox{\bf mylist}}
\def\addrel{\mbox{\bf addrel}}
\def\addrelz{\mbox{\tt addrel\_z}}
\def\addrels{\mbox{\tt addrel\_s}}
\begin%
{figure}
\begin%
{verbatim}
datatype Z() = Z of ()
datatype S(a:type) = S of a
//
datatype
mylist(type, type) =
  | {a:type}
    mynil(a, Z())
  | {a:type}{n:type}
    mycons(a, S(n)) of (a, mylist(a, n))
//
datatype
addrel(type, type, type) =
  | {n:type}
    addrel_z(Z(), n, n) of ()
  | {m,n:type}{r:type}
    addrel_s(S(m), n, S(r)) of addrel(m, n, r)
//
fun
myappend
{a:type}
{m,n:type}
(
  xs: mylist(a, m)
, ys: mylist(a, n)
) : [r:type]
(
  addrel(m, n, r), mylist(a, r)
) =
(
  case xs of
  | mynil() => let
      val pf = addrel_z() in (pf, ys)
    end // end of [mynil]
  | mycons(x, xs) => let
      val (pf, res) = myappend(xs, ys) in (addrel_s(pf), mycons(x, res))
    end // end of [mycons]
)
\end{verbatim}
\caption{A motivating example for PwTP in ATS}
\label{figure:mylist_append_1}
\end{figure}
The datatypes $\tZ$ and $\tS$ are declared in
Figure~\ref{figure:mylist_append_1} solely for representing natural
numbers: $\tZ$ represents $0$, and $\tS(N)$ represents the successor
of the natural number represented by $N$. The data constructors
associated with $\tZ$ and $\tS$ are of no use.  Given a type $T$ and
another type $N$, $\mylist(T, N)$ is a type for lists containing $n$
elements of the type $T$, where $n$ is the natural number represented
by $N$. Note that $\mylist$ is not a standard datatype (as is
supported in ML); it is a {\em guarded recursive datatype}
(GRDT)~\cite{GRDT-popl03}, which is also known as {\em generalized
  algebraic datatype} (GADT)~\cite{PhantomTypes} in Haskell and OCaml.
The datatype $\addrel$ is declared to capture the relation induced by
the addition function on natural numbers. Given types $M$, $N$, and
$R$ representing natural numbers $m$, $n$, and $r$, respectively, the
type $\addrel(M, N, R)$ is for a value representing some proof of
$m+n=r$. Note that $\addrel$ is also a GRDT or GADT. There are two
constructors $\addrelz$ and $\addrels$ associated with $\addrel$,
which encode the following two rules:
\[
\begin%
{array}{rcll}
0 + n & = & n & \mbox{for every natural number $n$} \\
(m+1) + n & = & (m+n)+1 & \mbox{for every pair of natural numbers $m$ and $n$} \\
\end{array}\]

\def\fmyappend{\mbox{\tt myappend}}
Let us now take a look at the implementation of $\fmyappend$.
Formally, the type assigned to $\fmyappend$ can be written as follows:
\[
\begin%
{array}{l}
\forall a:\stype.\forall m:\stype.\forall n:\stype. \\
~~~~(\mylist(a, m), \mylist(a, n)) \timp \exists r:\stype.~(\addrel(m, n, r), \mylist(a, r)) \\
\end{array}
\]
In essence, this type states the following: Given two lists of length
$m$ and $n$, $\fmyappend$ returns a pair such that the first component
of the pair is a proof showing that $m+n$ equals $r$ for some natural
number $r$ and the second component is a list of length $r$.

Unlike $\fappend$, type-checking $\fmyappend$ does not generate any
linear constraints on integers.  As a matter of fact, $\fmyappend$ can
be readily implemented in both Haskell and OCaml (extended with
support for generalized algebraic datatypes), where there is no
built-in support for handling linear constraints on integers.  This is
an example of great significance in the sense that it demonstrates
concretely an approach to allowing the programmer to write code of the
nature of theorem-proving so as to simplify or even eliminate certain
constraints that need otherwise to be solved directly during
type-checking.  With this approach, constraint-solving is effectively
internalized, and the programmer can actively participate in
constraint simplification, gaining a tight control in determining what
constraints should be passed to the underlying constraint-solver.

There are some major issues with the implementation given in
Figure~\ref{figure:mylist_append_1}. Clearly, representing natural
numbers as types is inadequate since there are types that do not
represent any natural numbers. More seriously, this representation
turns quantification over natural numbers (which is predicative) into
quantification over types (which is impredicative), causing
unnecessary complications. Also, proof construction (that is,
construction of values of types formed by $\addrel$) needs to be
actually performed at run-time, which causes inefficiency both
time-wise and memory-wise. Probably the most important issue is that
proof validity is not guaranteed. For instance, it is entirely
possible to fake proof construction by making use of non-terminating
functions.

\def\mynat{{\it mynat}}
\begin%
{figure}
\begin%
{verbatim}
datasort
mynat = Z of () | S of mynat
//
datatype
mylist(type, mynat) =
  | {a:type}
    mynil(a, Z())
  | {a:type}{n:mynat}
    mycons(a, S(n)) of (a, mylist(a, n))
//
dataprop
addrel(mynat, mynat, mynat) =
  | {y:mynat}
    addrel_z(Z, y, y) of ()
  | {x,y:mynat}{r:mynat}
    addrel_s(S(x), y, S(r)) of addrel(x, y, r)
//
fun
myappend
{a:type}
{m,n:mynat}
(
  xs: mylist(a, m)
, ys: mylist(a, n)
) : [r:mynat]
(
  addrel(m, n, r) | mylist(a, r)
) =
(
  case xs of
  | mynil() => let
      val pf = addrel_z() in (pf | ys)
    end // end of [mynil]
  | mycons(x, xs) => let
      val (pf | res) = myappend(xs, ys) in (addrel_s(pf) | mycons(x, res))
    end // end of [mycons]
)
\end{verbatim}
\caption{An example making use of PwTP in $\atslang$}
\label{figure:mylist_append_2}
\end{figure}
In Figure~\ref{figure:mylist_append_2}, another implementation of
$\fmyappend$ is given that makes use of the support for PwTP in
$\atslang$. Instead of representing natural numbers as types, a
datasort of the name $\mynat$ is declared and natural numbers can be
represented as static terms of the sort $\mynat$. Also, a dataprop
$\addrel$ is declared for capturing the relation induced by the
addition function on natural numbers. As a dataprop, $\addrel$ can
only form types for values representing proofs, which are erased after
type-checking and thus need no construction at run-time. In the
implementation of $\fmyappend$, the bar symbol ($\mbox{\tt |}$) is
used in place of the comma symbol to separate components in tuples;
the components appearing to the left of the bar symbol are proof
expressions (to be erased) and those to the right are dynamic
expressions (to be evaluated). After proof-erasure, the implementation
of $\fmyappend$ essentially matches that of $\fappend$ given in
Figure~\ref{figure:list_append_function}.

As a framework to facilitate the design and formalization of
advanced type systems for practical programming, $\ATS$ is first
formulated with no support for PwTP~\cite{ATS-types03}. This
formulation is the basis for a type system referred to as $\ATSzero$
in this paper. The support for PwTP is added into $\ATS$ in a
subsequent formulation~\cite{CPwTP-icfp05}, which serves as the basis
for a type system referred to as $\ATSproof$ in this paper. However, a
fundamentally different approach is adopted in $\ATSproof$ to justify
the soundness of PwTP, which essentially translates each well-typed
program in $\ATSproof$ into another well-typed one in $\ATSzero$ of
the same dynamic semantics.  The identification and formalization of
this approach, which is both simpler and more general than one used
previously~\cite{CPwTP-icfp05}, consists of a major technical
contribution of the paper.

It is intended that
the paper should focus on the theoretical development of $\ATS$,
and the presentation given is of a minimalist style.
The organization for the rest of the paper is given as follows. An
untyped $\lambda$-calculus $\lamdyn$ is first presented in
Section~\ref{section:lamdyn} for the purpose of introducing some basic
concepts needed to formally assign dynamic (that, operational)
semantics to programs. In Section~\ref{section:ATSzero}, a generic
applied type system $\ATSzero$ is formulated and its type-soundness
established.  Subsequently, $\ATSzero$ is extended to $\ATSproof$ in
Section~\ref{section:ATSproof} with support for PwTP, and the
type-soundness of $\ATSproof$ is reduced to that of $\ATSzero$ through
a translation from well-typed programs in the former to those in the
latter. Lastly, some closely related work is discussed in
Section~\ref{section:related} and the paper concludes.

\def\dcc{{\it dcc}}
\def\dcf{{\it dcf}}
\def\dcx{{\it dcx}}
\def\dsif{\mbox{\bf if}}
\def\dfst{\mbox{\bf fst}}
\def\dsnd{\mbox{\bf snd}}
\def\dapp#1#2{{\bf app}(#1,#2)}
\def\dcase#1#2{{\bf case}~#1~{\bf of}~#2}
\def\dsapp#1#2{{\bf sapp}(#1,#2)}
\def\dfix#1#2{{\bf fix}\;#1.#2}
\def\dlam#1#2{{\bf lam}\;#1.\kern1pt#2}
\def\dslam#1#2{{\bf slam}\;#1.\kern1pt#2}
\def\dletin#1#2{{\bf let}\;#1\;{\bf in}\;#2}
\def\dsletin#1#2{{\bf let}\;#1\;{\bf in}\;#2}
\def\dtuple#1{\langle #1\rangle}
\def\dstuple#1{\langle #1\rangle}
\def\deguard#1{\supset^{-}\kern-2pt(#1)}
\def\diguard#1{\supset^{+}\kern-2pt(#1)}
\def\dassert#1{\Band\kern-0.50pt(#1)}
\def\emptydsub{[]}
\def\subst#1#2#3{#3[#2\mapsto #1]}

\section%
{%
Untyped%
~$\lambda$-Calculus $\lamdyn$%
}\label{section:lamdyn}
The purpose of formulating $\lamdyn$, an untyped lambda-calculus
extended with constants (including constant constructors and constant
functions), is to set up some machinery needed to formalize dynamic
(that is, operational) semantics for programs. It is to be proven that
a well-typed program in $\ATS$ can be turned into one in $\lamdyn$
through type-erasure and proof-erasure while retaining its dynamic
semantics, stressing the point that types and proofs in $\ATS$ play no
active r{\^o}le in the evaluation of a program. In this regard, the
form of typing studied in $\ATS$ is of Curry-style (in contrast with
Church-style)~\cite{REYNOLDS-book98}.

There are no static terms in $\lamdyn$. The syntax for the dynamic
terms in $\lamdyn$ is given as follows:
\[%
\begin%
{array}{lrcl}
\mbox{dynamic terms} & e & ::= & %
x\mid\dcx(\vec{e})\mid
\dtuple{e_1,e_2}\mid\dfst(e)\mid\dsnd(e)\mid \\
& & & \dlam{x}{e}\mid\dapp{e_1}{e_2}\mid\dletin{x=e_1}{e_2} \\
\end{array}\]
where the notation $\vec{e}$ is for a possibly empty sequence of
dynamic terms.  Let $\dcx$ range over external dynamic constants,
which include both dynamic constructors $\dcc$ and dynamic functions
$\dcf$. The arguments taken by a dynamic constructor or function are
often primitive values (instead of those constructed by $\mbox{\bf
  lam}$ and $\dtuple{\cdot,\cdot}$) and the result returned by it is
often a primitive value as well.  The meaning of various forms of
dynamic terms should become clear when the rules for evaluating them
are given.

The values in $\lamdyn$ are just special forms of dynamic terms, and
the syntax for them is given as follows:
\[%
\begin%
{array}{lrcl}
\mbox{values} & v & ::= & x\mid\dcc(\vec{v})\mid\dtuple{v_1,v_2}\mid\dlam{x}{e} \\
\end{array}\]
where $\vec{v}$ is for a possibly empty sequence of values.
A standard approach to assigning dynamic semantics to terms is based
on the notion of evaluation contexts:
\[%
\begin%
{array}{lrcl}
\mbox{evaluation contexts} & E & ::= & %
[]\mid\dcx(v_1,\ldots,v_{i-1},E,e_{i+1},\ldots,e_n)\mid \\
& & & \dtuple{E,e}\mid\dtuple{v,E}\mid\dapp{E}{e}\mid\dapp{v}{E}\mid\dletin{x=E}{e} \\
\end{array}\]
Essentially, an evaluation context $E$ is a dynamic term in which a
subterm is replaced with a hole denoted by $[]$. Note that only
subterms at certain positions in a dynamic term can be replaced to
form valid evaluation contexts.

\def\eval{\rightarrow}
\def\meval{\rightarrow^{*}}
\begin%
{definition}
The redexes in $\lamdyn$ and their reducts are defined as follows:
\begin%
{itemize}
\item
$\dfst(\dtuple{v_1,v_2})$ is a redex, and its reduct is $v_1$.
\item
$\dsnd(\dtuple{v_1,v_2})$ is a redex, and its reduct is $v_2$.
\item
$\dapp{\dlam{x}{e}}{v}$ is a redex, and its reduct is $\subst{v}{x}{e}$.
\item
$\dcf(\vec{v})$ is a redex if it is defined to equal some value $v$;
if so, its reduct is $v$.
\end{itemize}
\end{definition}
Note that it may happen later that a new form of redex can have more
than one reducts. Given a dynamic term of the form $E[e_1]$ for some
redex $e_1$, $E[e_1]$ is said to reduce to $E[e_2]$ in one-step if
$e_2$ is a reduct of $e_1$, and this one-step reduction is denoted by
$E[e_1]\eval E[e_2]$.  Let $\meval$ stand for the reflexive and
transitive closure of $\eval$.

Given a program (that is, a closed dynamic term) $e_0$ in $\lamdyn$, a
finite reduction sequence starting from $e_0$ can either lead to a
value or a non-value. If a non-value cannot be further reduced, then
the non-value is said to be {\em stuck} or in a {\em stuck} form. In
practice, values can often be represented in special manners to allow
various stuck forms to be detected through checks performed at
run-time. For instance, the representation of a value in a dynamically
typed language most likely contains a tag to indicate the type of the
value. If it is detected that the evaluation of a program reaches a
stuck form, then the evaluation can be terminated abnormally with a
raised exception.

Detecting potential stuck forms that may occur during the evaluation
of a program can also be done statically (that is, at
compiler-time). One often imposes a type discipline to ensure the
absence of various stuck forms during the evaluation of a well-typed
program. This is the line of study to be carried out in the rest of
the paper.

\section%
{%
Formal Development of $\ATSzero$%
}\label{section:ATSzero}

As a generic applied type system, $\ATSzero$ consists of a static
component (statics), where types are formed and reasoned about, and a
dynamic component (dynamics), where programs are constructed and
evaluated.  The statics itself is a simply typed lambda-calculus
(extended with certain constants), and the types in it are called {\em
  sorts} so as to avoid confusion with the types for classifying
dynamic terms, which are themselves static terms.

\def\basesort{b}
\def\simp{\rightarrow}
\def\Simp{\Rightarrow}
\def\scc{\mbox{\it scc}}
\def\scf{\mbox{\it scf}}
\def\scx{\mbox{\it scx}}
\def\emptyssub{[]}
\def\ttrue{\mbox{\it true}}
\def\ffalse{\mbox{\it false}}
\begin{figure}
\[\begin{array}{lrcl}
\mbox{sorts} %
   & \sigma & ::= & \basesort \mid\sigma_1\simp\sigma_2 \\
\mbox{static terms} %
   & s & ::= & a \mid \scx(s_1,\ldots,s_n) \mid \lambda a:\sigma.s \mid s_1(s_2) \\
\mbox{static var. ctx.} %
   & \Sigma & ::= & \emptyset \mid \Sigma, a:\sigma \\
\mbox{static subst.} %
   & \Theta & ::= & \emptyssub \mid \Theta[a\mapsto s] \\
\end{array}\]
\caption{The syntax for the statics of $\ATSzero$}
\label{figure:syntax_for_ATS0_statics}
\end{figure}
The syntax for the statics of $\ATSzero$ is given in
Figure~\ref{figure:syntax_for_ATS0_statics}.  Let $\basesort$ range
over the base sorts in $\ATSzero$, which include at least $\sbool$ for
static booleans and $\stype$ for types (assigned to dynamic
terms). The base sort $\sint$ for static integers is not really needed
for formalizing $\ATSzero$ but it is often used in the presented
examples. Let $a$ and $s$ range over static variables and static
terms, respectively.  There may be some built-in static constants
$\scx$, which are either static constant constructors $\scc$ or static
constant functions $\scf$.  A c-sort is of the form
$(\sigma_1,\ldots,\sigma_n)\Simp b$, which can only be assigned to
static constants. Note that a c-sort is not considered a (regular)
sort.  Given a static constant $\scx$, a static term
$\scx(s_1,\ldots,s_n)$ is of sort $b$ if $\scx$ is assigned a c-sort
$(\sigma_1,\ldots,\sigma_n)\Simp b$ for some sorts
$\sigma_1,\ldots,\sigma_n$ and $s_i$ can be assigned the sorts
$\sigma_i$ for $i=1,\ldots,n$.  It is allowed to write $\scc$ for
$\scc()$ if there is no risk of confusion.  In $\ATSzero$, the
existence of the following static constants with the assigned c-sorts
is assumed:
\def\Band{\land}
\def\Bimp{\supset}
\def\tyleq{\leq_{ty}}
\def\tand{*}
\def\timp{\rightarrow}
\def\Timp{\Rightarrow}
\def\tunit{{\mathbf 1}}
\def\tuple#1{\langle #1\rangle}
\[\begin{array}{ccl}
\ttrue %
      & : & ()\Simp \sbool \\
\ffalse %
      & : & ()\Simp \sbool \\
\tyleq %
      & : & (\stype,\stype) \Simp \sbool \\
\tand & : & (\stype,\stype) \Simp \stype \\
\timp & : & (\stype,\stype) \Simp \stype \\
\Band & : & (\sbool,\stype) \Simp \stype \\
\Bimp & : & (\sbool,\stype) \Simp \stype \\
\forall_\sigma & : & (\sigma\simp\stype)\Simp\stype \\
\exists_\sigma & : & (\sigma\simp\stype)\Simp\stype \\
\end{array}\]
Note that infix notation may be used for certain static constants. For
instance, $s_1\timp s_2$ stands for $\timp(s_1,s_2)$ and $s_1\tyleq
s_2$ stands for $\tyleq(s_1,s_2)$.  In addition, $\forall a:\sigma.s$
and $\exists a:\sigma.s$ stand for $\forall_\sigma(\lambda
a:\sigma.s)$ and $\exists_\sigma(\lambda a:\sigma.s)$, respectively.
Given a static constant constructor $\scc$, if the c-sort assigned to
$\scc$ is $(\sigma_1,\ldots,\sigma_n)\Simp\stype$ for some sorts
$\sigma_1,\ldots,\sigma_n$, then $\scc$ is a type constructor.  For
instance, $\tand$, $\timp$, $\Band$, $\Bimp$, $\forall_{\sigma}$ and
$\exists_{\sigma}$ are all type constructors.  Additional built-in base
type constructors may be assumed.

\def\tInt{\mbox{\bf Int}}
\def\tNat{\mbox{\bf Nat}}
\def\tbool{\mbox{\bf bool}}
\def\tBool{\mbox{\bf Bool}}
Given a proposition $B$ and a type $T$, $B\Bimp T$ is a guarded type and
$B\Band T$ is an asserting type. Intuitively, if a value $v$ is assigned a
guarded type $B\Bimp T$, then $v$ can be used only if the guard $B$ is
satisfied; if a value $v$ of an asserting type $B\Band T$ is generated at a
program point, then the assertion $B$ holds at that point. For instance,
suppose that $\sint$ is a sort for (static) integers and $\tint$ is a type
constructor of the sort $(\sint)\Simp\stype$; given a static term $s$ of
the sort $\sint$, $\tint(s)$ is a singleton type for the integer equal to
$s$; hence, the usual type $\tInt$ for (dynamic) integers can be defined as
$\exists a:\sint.~\tint(a)$, and the type $\tNat$ for natural numbers can
be defined as $\exists a:\sint.~(a\geq 0)\Band\tint(a)$.  Moreover, the
following type is for the (dynamic) division function on integers:
\[%
\begin%
{array}{c}
\forall a_1:\sint.\forall a_2:\sint.~a_2\not=0\Bimp(\tint(a_1),\tint(a_2))\timp\tint(a_1/a_2) \\
\end{array}\]
where the meaning of $\not=$ and $/$ should be obvious.  With such a type,
division by zero is disallowed during type-checking (at
compile-time). Also, suppose that $\tbool$ is a type constructor of the
sort $(\sbool)\Simp\stype$ such that for each proposition $B$, $\tbool(B)$
is a singleton type for the truth value equal to $B$. Then the usual type
$\tBool$ for (dynamic) booleans can be defined as $\exists a:\sbool.~\tbool(a)$.
The following type is an interesting one:
$$\forall a:\sbool.~\tbool(a)\timp a\Band\tunit$$ where $\tunit$
stands for the unit type. Given a function $f$ of this type, we can
apply $f$ to a boolean value $v$ of type $\tbool(B)$ for some
proposition $B$; if $f(v)$ returns, the $B$ must be true; therefore
$f$ acts like dynamic assertion-checking.

For those familiar with qualified types~\cite{QualifiedTypes}, which
underlies the type class mechanism in Haskell, it should be noted that
a qualified type cannot be regarded as a guarded type.  The simple
reason is that the proof of a guard in $\ATSzero$ bears no
computational significance, that is, it cannot affect the run-time behavior
of a program, while a dictionary, which is just a proof of some
predicate on types in the setting of qualified types, can and is
mostly likely to affect the run-time behavior of a program.

\def\strule{st}
\def\tyrule{ty}
\def\regrule{reg}
\def\tpjg{\vdash}
\def\temd{\models}
\begin{figure}
\[\begin{array}{c}
\infer[\mbox{\bf(\strule-var)}]
      {\Sigma \tpjg a:\sigma}
      {\Sigma(a) = \sigma} \\[4pt]

\infer[\mbox{\bf(\strule-scx)}]
      {\Sigma\tpjg\scx(s_1,\ldots,s_n): b}
      {\tpjg\scx: (\sigma_1,\ldots,\sigma_n)\Simp b &
       \Sigma\tpjg s_1:\sigma_1 & \cdots & \Sigma\tpjg s_n:\sigma_n} \\[4pt]

\infer[\mbox{\bf(\strule-lam)}]
      {\Sigma \tpjg \lambda a:\sigma_1.s: \sigma_1\simp\sigma_2}
      {\Sigma, a:\sigma_1 \tpjg s:\sigma_2} \\[4pt]

\infer[\mbox{\bf(\strule-app)}]
      {\Sigma \tpjg s_1(s_2):\sigma_2}
      {\Sigma \tpjg s_1:\sigma_1\simp\sigma_2 &
       \Sigma \tpjg s_2:\sigma_1} \\[4pt]

\end{array}\]
\caption{The sorting rules for the statics of $\ATSzero$}
\label{figure:ATS0_sorting_rules}
\end{figure}
\def\vB{\vec{B}}
\def\dom{\mbox{\bf dom}}
The standard rules for assigning sorts to static terms are given in
Figure~\ref{figure:ATS0_sorting_rules}, where the judgement $\tpjg
\scx: (\sigma_1,\ldots,\sigma_n)\Simp b$ means that the static
constant $\scx$ is assumed to be of the c-sort
$(\sigma_1,\ldots,\sigma_n)\Simp b$.  Given $\vec{s}=s_1,\ldots,s_n$
and $\vec{\sigma}=\sigma_1,\ldots,\sigma_n$, a judgement of the form
$\Sigma\tpjg \vec{s}:\vec{\sigma}$ means $\Sigma\tpjg s_i:\sigma_i$
for $i=1,\ldots,n$.  Let $B$ stand for a static term that can be
assigned the sort $\sbool$ (under some context $\Sigma$) and $\vB$ a
possibly empty sequence of static boolean terms. Also, let $T$ stand
for a type (for dynamic terms), which is a static term that can be
assigned the sort $\stype$ (under some context $\Sigma$). Given
contexts $\Sigma_1$ and $\Sigma_2$ and a substitution $\Theta$, the
judgement $\Sigma_1\tpjg\Theta:\Sigma_2$ means that
$\Sigma_1\tpjg\Theta(a):\Sigma_2(a)$ is derivable for each
$a\in\dom(\Theta)=\dom(\Sigma_2)$.

\begin%
{proposition}\label{prop:static-subst}
Assume $\Sigma\tpjg s:\sigma$ is derivable.  If
$\Sigma=\Sigma_1,\Sigma_2$ and $\Sigma_1\tpjg\Theta:\Sigma_2$ holds,
then $\Sigma_1\tpjg s[\Theta]:\sigma$ is derivable.
\end{proposition}
\begin%
{proof}
By structural induction on the derivation of $\Sigma\tpjg s:\sigma$.
\hfill
\end{proof}

\begin{figure}[thp]
\[\begin{array}{c}
\infer[\mbox{\bf(\regrule-id)}]
      {\Sigma;\vB\temd B}
      {B\in\vB} \\[4pt]

\infer[\mbox{\bf(\regrule-true)}]
      {\Sigma;\vB\temd\ttrue}
      {} \\[4pt]

\infer[\mbox{\bf(\regrule-false)}]
      {\Sigma;\vB\temd B}
      {\Sigma;\vB\temd\ffalse} \\[4pt]




\infer[\mbox{\bf(\regrule-var-thin)}]
      {\Sigma, a:\sigma;\vB \temd B}
      {\Sigma;\vB \temd B} \\[4pt]

\infer[\mbox{\bf(\regrule-bool-thin)}]
      {\Sigma;\vB,B_1 \temd B_2}
      {\Sigma\tpjg B_1:\sbool & \Sigma;\vB \temd B_2} \\[4pt]

\infer[\mbox{\bf(\regrule-subst)}]
      {\Sigma; \subst{s}{a}{\vB} \temd \subst{s}{a}{B}}
      {\Sigma, a:\sigma; \vB \temd B & \Sigma\tpjg s:\sigma} \\[4pt]

\infer[\mbox{\bf(\regrule-cut)}]
      {\Sigma;\vB \temd B_2}
      {\Sigma; \vB \temd B_1 & \Sigma;\vB, B_1 \temd B_2} \\[4pt]

\end{array}\]
\caption{The regularity rules for the constraint relation in $\ATSzero$}
\label{figure:ATS0_regularity_rules}
\end{figure}

\begin%
{definition}%
[Constraints in $\ATSzero$]
\label{ATS0_constraint_def}
A constraint in $\ATSzero$ is of the form $\Sigma;\vB\temd B_0$,
where $\Sigma\tpjg B:\sbool$ holds for each $B$ in $\vB$ and
$\Sigma\tpjg B_0:\sbool$ holds as well, and the constraint relation in
$\ATSzero$ is the one that determines whether each constraint is true
or false.

Each regularity rule in Figure~\ref{figure:ATS0_regularity_rules} is
assumed to be met, that is, the conclusion of each regularity rule
holds if all of its premisses hold, and the following regularity
conditions on $\tyleq$ are also satisfied:
\begin{enumerate}
\item
$\Sigma;\vB\temd T\tyleq T$ holds for every $T$.
\item
$\Sigma;\vB\temd T\tyleq T'$ and $\Sigma;\vB\temd T'\tyleq T''$ implies
$\Sigma;\vB\temd T\tyleq T''$.
\item
$\Sigma;\vB\temd T_1\tand T_2\tyleq T'_1\tand T'_2$ implies
$\Sigma;\vB\temd T_1\tyleq T'_1$ and $\Sigma;\vB\temd T_2\tyleq T'_2$.
\item
$\Sigma;\vB\temd T_1\timp T_2\tyleq T'_1\timp T'_2$ implies
$\Sigma;\vB\temd T'_1\tyleq T_1$ and $\Sigma;\vB\temd T_2\tyleq T'_2$.
\item
$\Sigma;\vB\temd B\Band T\tyleq B'\Band T'$ implies
$\Sigma;\vB,B\temd B'$ and $\Sigma;\vB,B\temd T\tyleq T'$.
\item
$\Sigma;\vB\temd B\Bimp T\tyleq B'\Bimp T'$ implies
$\Sigma;\vB,B'\temd B$ and $\Sigma;\vB,B'\temd T\tyleq T'$.
\item
$\Sigma;\vB\temd \forall a:\sigma.T\tyleq\forall a:\sigma.T'$ implies
$\Sigma, a:\sigma;\vB\temd T\tyleq T'$.
\item
$\Sigma;\vB\temd \exists a:\sigma.T\tyleq\exists a:\sigma.T'$ implies
$\Sigma, a:\sigma;\vB\temd T\tyleq T'$.
\item
$\emptyset;\emptyset\temd\scc(T_1,\ldots,T_n)\tyleq T'$ implies $T'=\scc(T'_1,\ldots,T'_{n})$
for some $T'_1, \ldots, T'_n$.
\end{enumerate}
\end{definition}
The need for these conditions is to become clear when proofs are
constructed in the following presentation for formally establishing
various meta-properties of $\ATSzero$. For instance, the last of the
above conditions can be invoked to make the claim that $T'\tyleq
T_1\timp T_2$ implies $T'$ being of the form $T'_1\timp T'_2$. Note
that this condition actually implies the consistency of the constraint
relation as not every constraint is valid.

\begin%
{figure}
\[%
\begin%
{array}{lrcl}
\mbox{dynamic terms} %
                     & e & ::= & %
                           x \mid \dcx\{\vec{s}\}(e_1,\ldots,e_n) \mid \\
                     & & & \dtuple{e_1,e_2} \mid \dfst(e) \mid \dsnd(e) \mid \dlam{x}{e} \mid \dapp{e_1}{e_2} \mid \\
                     & & & \diguard{e} \mid \;\deguard{e} \mid \dslam{a}{e} \mid \dsapp{e}{s} \mid \\
                     & & & \dassert{e} \mid \dletin{\dassert{x}=e_1}{e_2} \mid \dstuple{s,e} \mid \dsletin{\dstuple{a,x}=e_1}{e_2} \\
\mbox{dynamic values} %
                     & v & ::= & x \mid \dcc\{\vec{s}\}(v_1,\ldots,v_n) \mid \\
                     & & & \dtuple{v_1,v_2} \mid \dlam{x}{e} \mid \diguard{e} \mid \dslam{a}{e} \mid \dassert{v} \mid \dstuple{s,v} \\
\mbox{dynamic var. ctx.} %
                     & \Delta & ::= & \emptyset \mid \Delta, x:T \\[6pt]
\mbox{dynamic subst.} %
                     & \Theta & ::= & \emptydsub \mid \Theta[x\mapsto e] \\[6pt]
\end{array}\]
\caption{The syntax for the dynamics in $\ATSzero$}
\label{figure:syntax_for_ATS0_dynamics}
\end{figure}
Let us now move onto the dynamic component (dynamics) of $\ATSzero$.
The syntax for the dynamics of $\ATSzero$ is given in
Figure~\ref{figure:syntax_for_ATS0_dynamics}.  Let $x$ range over
dynamic variables and $\dcx$ dynamic constants, which include both
dynamic constant constructors $\dcc$ and dynamic constant functions
$\dcf$.  Some (unfamiliar) forms of dynamic terms are to be understood
when the rules for assigning types to them are presented. Let $v$
range over values, which are dynamic terms of certain special forms,
and $\Delta$ range over dynamic variable contexts, which assign types
to dynamic variables.

\def\Der{{\cal D}}
\def\height#1{\mbox{\it ht}(#1)}
During the formal development of $\ATSzero$, proofs are often constructed
by induction on derivations (represented as trees).  Given a judgement $J$,
$\Der::J$ means that $\Der$ is a derivation of $J$, that is, the conclusion
of $\Der$ is $J$.  Given a derivation $\Der$, $\height{\Der}$ stands for
the height of the tree that represents $\Der$.

\begin{figure}[thp]
\fontsize{10}{11}\selectfont
\[\begin{array}{c}
\infer[\mbox{\bf(\tyrule-var)}]
      {\Sigma; \vB; \Delta \tpjg x:T}
      {\tpjg\Sigma; \vB; \Delta & \Delta(x) = T} \\[2pt]

\infer[\mbox{\bf(\tyrule-sub)}]
      {\Sigma; \vB; \Delta \tpjg e:T'}
      {\Sigma; \vB; \Delta \tpjg e:T & \Sigma;\vB \temd T\tyleq T'} \\[2pt]

\infer[\mbox{\bf(\tyrule-dcx)}]
      {\Sigma; \vB; \Delta \tpjg \dcx\{\vec{s}\}(e_1,\ldots,e_n):\subst{\vec{s}}{\vec{a}}{T}}
      {$$\begin{array}{c}
       \tpjg\dcx: \forall\vec{a}:\vec{\sigma}.\vB_0\Bimp(T_1,\ldots,T_n)\Timp T \\
       \Sigma\tpjg \vec{s}:\vec{\sigma} \kern18pt
       \Sigma;\vB \temd \subst{\vec{s}}{\vec{a}}{B}~~\mbox{for each $B\in\vB_0$} \\
       \Sigma;\vB;\Delta \tpjg e_i:\subst{\vec{s}}{\vec{a}}{T_i}\kern6pt\mbox{for $i=1,\ldots,n$} \\
       \end{array}$$} \\[2pt]

\infer[\mbox{\bf(\tyrule-tup)}]
      {\Sigma; \vB; \Delta\tpjg \dtuple{e_1,e_2}:T_1 \tand T_2}
      {\Sigma; \vB; \Delta\tpjg e_1: T_1 & \Sigma; \vB; \Delta\tpjg e_2: T_2} \\[2pt]

\infer[\mbox{\bf(\tyrule-fst)}]
      {\Sigma; \vB; \Delta\tpjg \dfst(e):T_1}
      {\Sigma; \vB; \Delta\tpjg e: T_1 * T_2}
\kern22pt
\infer[\mbox{\bf(\tyrule-snd)}]
      {\Sigma; \vB; \Delta\tpjg \dsnd(e):T_2}
      {\Sigma; \vB; \Delta\tpjg e: T_1 \tand T_2} \\[2pt]

\infer[\mbox{\bf(\tyrule-lam)}]
      {\Sigma; \vB; \Delta\tpjg \dlam{x}{e}: T_1\timp T_2}
      {\Sigma; \vB; \Delta, x: T_1\tpjg e: T_2} \\[2pt]

\infer[\mbox{\bf(\tyrule-app)}]
      {\Sigma; \vB; \Delta\tpjg \dapp{e_1}{e_2}: T_2}
      {\Sigma; \vB; \Delta\tpjg e_1: T_1\timp T_2 &
       \Sigma; \vB; \Delta\tpjg e_2: T_1} \\[2pt]

\infer[\mbox{\bf(\tyrule-$\Bimp$-intr)}]
      {\Sigma; \vB; \Delta \tpjg\; \diguard{e}: B\Bimp T}
      {\Sigma; \vB, B; \Delta \tpjg e: T} \\[2pt]

\infer[\mbox{\bf(\tyrule-$\Bimp$-elim)}]
      {\Sigma; \vB; \Delta \tpjg\; \deguard{e}: T}
      {\Sigma; \vB; \Delta \tpjg e: B\Bimp T & \Sigma; \vB\temd B} \\[2pt]

\infer[\mbox{\bf(\tyrule-$\Band$-intr)}]
      {\Sigma; \vB; \Delta \tpjg \Band(e): B\Band T}
      {\Sigma; \vB \temd B & \Sigma; \vB; \Delta \tpjg e: T} \\[2pt]

\infer[\mbox{\bf(\tyrule-$\Band$-elim)}]
      {\Sigma; \vB; \Delta \tpjg \dsletin{\Band(x)=e_1}{e_2}: T_2}
      {\Sigma; \vB; \Delta \tpjg e_1: B\land T_1 & \Sigma; \vB,B; \Delta, x:T_1 \tpjg e_2: T_2} \\[2pt]

\infer[\mbox{\bf(\tyrule-$\forall$-intr)}]
      {\Sigma; \vB; \Delta \tpjg \dslam{a}{e}: \forall a:\sigma.T}
      {\Sigma, a:\sigma; \vB; \Delta \tpjg e: T} \\[2pt]

\infer[\mbox{\bf(\tyrule-$\forall$-elim)}]
      {\Sigma; \vB; \Delta \tpjg \dsapp{e}{s}: \subst{s}{a}{T}}
      {\Sigma; \vB; \Delta \tpjg e: \forall a:\sigma.T &
       \Sigma \tpjg s:\sigma} \\[2pt]

\infer[\mbox{\bf(\tyrule-$\exists$-intr)}]
      {\Sigma; \vB; \Delta \tpjg \dstuple{s,d}: \exists a:\sigma.T}
      {\Sigma \tpjg s:\sigma &
       \Sigma; \vB; \Delta \tpjg e: \subst{s}{a}{T}} \\[2pt]

\infer[\mbox{\bf(\tyrule-$\exists$-elim)}]
      {\Sigma; \vB; \Delta \tpjg \dsletin{\dstuple{a,x}=e_1}{e_2}: T_2}
      {\Sigma; \vB; \Delta \tpjg e_1: \exists a:\sigma.T_1 &
       \Sigma, a:\sigma; \vB; \Delta, x:T_1 \tpjg e_2: T_2} \\[2pt]
\end{array}\]
\caption{The typing rules for the dynamics of $\ATSzero$}
\label{figure:typing_rules_for_ATS0_dynamics}
\end{figure}
In $\ATSzero$, a typing judgement is of the form $\Sigma;\vB;\Delta\tpjg
e:T$, and the rules for deriving such a judgement are given in
Figure~\ref{figure:typing_rules_for_ATS0_dynamics}.  Note that certain
obvious side conditions associated with some of the typing rules are
omitted for the sake of brevity. For instance, the variable $a$ is not
allowed to have free occurrences in $\vB$, $\Delta$, or $T$ when the rule
$\mbox{\bf(\tyrule-$\forall$-intr)}$ is applied.

Given $\vB=B_1,\ldots,B_n$, $\vB\Bimp T$ stands for
$B_1\Bimp(\cdots(B_n\Bimp T)\cdots)$.  Given $\vec{a}=a_1,\ldots,a_n$
and $\vec{\sigma}=\sigma_1,\ldots,\sigma_n$,
$\forall\vec{a}:\vec{\sigma}$ stands for the sequence of quantifiers:
$\forall a:\sigma_1.\cdots\forall a:\sigma_n$.  A c-type in $\ATSzero$
is of the form $\forall\vec{a}:\vec{\sigma}.~\vB\Bimp
(T_1,\ldots,T_n)\Timp T$.

The notation
$\tpjg\dcx:\forall\vec{a}:\vec{\sigma}.~\vB\Bimp(T_1,\ldots,T_n)\Timp T$
means that $\dcx$ is assumed to have the c-type following it; if
$\dcx$ is a constructor $\dcc$, then $T$ is assumed to be constructed
by some $\scc$ and $\dcc$ is said to be associated with $\scc$. For
instance, the list constructors and the integer addition and division
functions can be given the following c-types:
\[%
\begin%
{array}{ccl}
\cnil &~:~& \forall a:\stype.~\tlist (a, 0) \\
\ccons &~:~& \forall a:\stype.\forall n:\sint.~n\geq 0\Bimp(a, \tlist (a, n))\timp\tlist(a, n+1) \\
\mbox{\tt iadd} & : & \forall a_1:\sint.\forall a_2:\sint.~(\tint(a_1), \tint(a_2))\Timp\tint(a_1+a_2) \\
\mbox{\tt isub} & : & \forall a_1:\sint.\forall a_2:\sint.~(\tint(a_1), \tint(a_2))\Timp\tint(a_1-a_2) \\
\mbox{\tt imul} & : & \forall a_1:\sint.\forall a_2:\sint.~(\tint(a_1), \tint(a_2))\Timp\tint(a_1\kern1.25pt*\kern1.25pta_2) \\
\mbox{\tt idiv} & : & \forall a_1:\sint.\forall a_2:\sint.~a_2\neq 0\Bimp(\tint(a_1), \tint(a_2))\Timp\tint(a_1/a_2) \\
\end{array}\]
where the type constructors $\tint$ and $\tlist$ are type constructors
of the c-sorts $(\sint)\Simp\stype$ and $(\stype, \sint)\Simp\stype$,
respectively, and $+$, $-$, $*$, and $/$ are static constant functions
of the c-sort $(\sint,\sint)\Simp\sint$.

For a technical reason, the rule $\mbox{\bf(\tyrule-var)}$ is to be
replaced with the following one:
\[\begin{array}{c}
\infer[\mbox{\bf(\tyrule-var')}]
      {\Sigma; \vB; \Delta \tpjg x:T'}
      {\Delta(x) = T & \Sigma;\vB\temd T\tyleq T'} \\[6pt]

\end{array}\]
which combines $\mbox{\bf(\tyrule-var)}$ with
$\mbox{\bf(\tyrule-sub)}$.  This replacement is needed for establishing
the following lemma:
\begin%
{lemma}%
\label{lemma:var_eq}
Assume $\Der::\Sigma;\vB;\Delta,x:T_1\tpjg e:T_2$
and $\Sigma;\vB\temd T'_1\tyleq T_1$. Then there is a derivation
$\Der'$ for the typing judgement $\Sigma;\vB;\Delta,x:T'_1\tpjg e:T_2$
such that $\height{\Der'}=\height{\Der}$.
\end{lemma}
\begin{proof}
The proof follows from structural induction on $\Der$ immediately.
The only interesting case is the one where the last applied rule is
$\mbox{\bf(\tyrule-var')}$, and this case can be handled by simply
merging two consecutive applications of the rule
$\mbox{\bf(\tyrule-var')}$ into one (with the help of the regularity
condition stating that $\tyleq$ is transitive).
\hfill
\end{proof}

Given $\Sigma,\vB,\Delta_1,\Delta_2$ and $\theta$, the judgement
$\Sigma;\vB;\Delta_1\tpjg\theta:\Delta_2$ means that the typing judgement
$\Sigma;\vB;\Delta_1\tpjg\theta(x):\Delta_2(x)$ is derivable for each
$x\in\dom(\theta)=\dom(\Delta_2)$.
\begin%
{lemma}%
[Substitution in $\ATSzero$]\label{lemma:substitution}
Assume $\Der::\Sigma;\vB;\Delta\tpjg e:T$ in $\ATSzero$.
\begin{enumerate}
\item
If $\vB=\vB_1,\vB_2$ and $\Sigma;\vB_1\temd\vB_2$ holds, then
$\Sigma;\vB_1;\Delta\tpjg e:T$ is also derivable, where
$\Sigma;\vB_1\temd\vB_2$ means $\Sigma;\vB_1\temd B$ holds for each
$B\in\vB_2$.
\item
If $\Sigma=\Sigma_1,\Sigma_2$ and $\Sigma_1\tpjg\Theta:\Sigma_2$ holds,
then $\Sigma_1;\vB[\Theta];\Delta[\Theta]\tpjg d[\Theta]:T[\Theta]$ is
also derivable.
\item
If $\Delta=\Delta_1,\Delta_2$ and $\Sigma;\vB;\Delta_1\tpjg\theta:\Delta_2$ is
derivable, then $\Sigma;\vB;\Delta_1\tpjg d[\theta]:T$ is also
derivable.
\end{enumerate}
\end{lemma}
\begin{proof}
By structural induction on the derivation $\Der$.
\hfill
\end{proof}

\begin%
{lemma}%
[Canonical Forms]\label{lemma:canonical_forms}
Assume $\Der::\emptyset;\emptyset;\emptyset\tpjg v:T$. Then the
following statements hold:
\begin{enumerate}
\item If $T=T_1\tand T_2$, then $v$ is of the form $\dtuple{v_1,v_2}$.
\item If $T=T_1\timp T_2$, then $v$ is of the form $\dlam{x}{e}$.
\item If $T=B\Band T_0$, then $v$ is of the form $\dassert{v_0}$.
\item If $T=B\Bimp T_0$, then $v$ is of the form $\diguard{e}$.
\item If $T=\forall a:\sigma.T_0$, then $v$ is of the form $\dslam{a}{e}$.
\item If $T=\exists a:\sigma.T_0$, then $v$ is of the form $\dstuple{s,v_0}$.
\item
If $T=\scc(\vec{s}_1)$, then $v$ is of the form $\dcc\{\vec{s}_2\}(\vec{v})$
for some $\dcc$ associated with $\scc$.
\end{enumerate}
\end{lemma}
\begin{proof}
With Definition~\ref{ATS0_constraint_def}, the lemma follows from
structural induction on $\Der$.  If the last applied rule in $\Der$ is
$\mbox{\bf(\tyrule-sub)}$, then the proof goes through by invoking the
induction hypothesis on the immediate subderivation of
$\Der$. Otherwise, the proof follows from a careful inspection of the
typing rules in Figure~\ref{figure:typing_rules_for_ATS0_dynamics}.
\hfill
\end{proof}

In order to assign (call-by-value) dynamic semantics to the dynamic
terms in $\ATSzero$, let us introduce evaluation contexts as follows:
\[\begin{array}{lrcl}
\mbox{eval. ctx.} & E & ::= & \hbox to 212pt{\hss} \\
~~~~~~\hbox to 0pt{$[] \mid \dcx\{\vec{s}\}(\vec{v},E,\vec{e}) \mid \dtuple{E,d} \mid \dtuple{v,E} \mid \dapp{E}{e} \mid \dapp{v}{E} \mid$\hss} &&&\\
~~~~~~\hbox to 0pt{$\deguard{E} \mid \dassert{E} \mid \dletin{\dassert{x}=E}{e} \mid \dsapp{E}{s} \mid \dstuple{s,E} \mid \dsletin{\dstuple{a,x}=E}{e}$\hss} &&&\\
\end{array}\]
\begin%
{definition}\label{ATS0_redex_def}
The redexes and their reducts are defined as follows.
\begin%
{itemize}
\item
$\dfst(\dtuple{v_1,v_2})$ is a redex, and its reduct is $v_1$.
\item
$\dsnd(\dtuple{v_1,v_2})$ is a redex, and its reduct is $v_2$.
\item
$\dapp{\dlam{x}{e}}{v}$ is a redex, and its reduct is $\subst{v}{x}{e}$.
\item
$\dcf\{\vec{s}\}(\vec{v})$ is a redex if it is defined to equal some
value $v$; if so, its reduct is $v$.
\item
$\deguard{\diguard{e}}$ is a redex, and its reduct is $e$.
\item
$\dsapp{\dslam{a}{e}}{s}$ is a redex, and its reduct is $\subst{s}{a}{e}$.
\item
$\dletin{\dassert{x}=\dassert{v}}{e}$ is a redex, and its reduct is
$\subst{v}{x}{e}$.
\item
$\dsletin{\dstuple{a,x}=\dstuple{s,v}}{e}$ is a redex, and its reduct is
$\subst{v}{x}{\subst{s}{a}{e}}$.
\end{itemize}
Given two dynamic terms $e_1$ and $e_2$ such that $e_1=E[e]$ and
$e_2=E[e']$ for some redex $e$ and its reduct $e'$, $e_1$ is said to reduce
to $e_2$ in one step and this one-step reduction is denoted by $e_1\eval
e_2$. Let $\meval$ stand for the reflexive and transitive closure of
$\eval$.
\end{definition}
It is assumed that the type assigned to each dynamic constant function
$\dcf$ is appropriate, that is, $\emptyset;\emptyset;\emptyset\tpjg
v:T$ is derivable whenever
$\emptyset;\emptyset;\emptyset\tpjg\dcf\{\vec{s}\}(v_1,\ldots,v_n):T$
is derivable and $v$ is a reduct of $\dcf\{\vec{s}\}(v_1,\ldots,v_n)$.

\begin%
{lemma}%
[Inversion]\label{lemma:inversion}
Assume $\Der::\Sigma;\vB;\Delta\tpjg e:T$ in $\ATSzero$.
\begin{enumerate}


\item
If $e=\dtuple{e_1,e_2}$,
then there exists $\Der'::\Sigma;\vB;\Delta\tpjg e:T$ such that
$\height{\Der'}\leq\height{\Der}$ and the last rule applied in $\Der'$ is
$\mbox{\bf(\tyrule-tup)}$.

\item
If $e=\dlam{x}{e_1}$,
then there exists $\Der'::\Sigma;\vB;\Delta\tpjg e:T$ such that
$\height{\Der'}\leq\height{\Der}$ and the last applied rule in $\Der'$ is
$\mbox{\bf(\tyrule-lam)}$.

\item
If $e=\diguard{e_1}$,
then there exists $\Der'::\Sigma;\vB;\Delta\tpjg e:T$ such that
$\height{\Der'}\leq\height{\Der}$ and the last rule applied in $\Der'$ is
$\mbox{\bf(\tyrule-$\Bimp$-intr)}$.

\item
If $e=\dassert{e_1}$,
then there exists $\Der'::\Sigma;\vB;\Delta\tpjg e:T$ such that
$\height{\Der'}\leq\height{\Der}$ and the last rule applied in $\Der'$ is
$\mbox{\bf(\tyrule-$\Band$-intr)}$.

\item
If $e=\dslam{a}{e_1}$,
then there exists $\Der'::\Sigma;\vB;\Delta\tpjg e:T$ such that
$\height{\Der'}\leq\height{\Der}$, and the last rule applied in $\Der'$ is
$\mbox{\bf(\tyrule-$\forall$-intr)}$.

\item
If $e=\dstuple{s,e_1}$,
then there exists $\Der'::\Sigma;\vB;\Delta\tpjg e:T$ such that
$\height{\Der'}\leq\height{\Der}$, and the last rule applied in $\Der'$ is
$\mbox{\bf(\tyrule-$\exists$-intr)}$.

\end{enumerate}
\end{lemma}
\begin{proof}
Let $\Der'$ be $\Der$ if
$\Der$ does not end with an application of the rule
$\mbox{\bf(\tyrule-sub)}$. Hence, in the rest of the proof, it can be assumed
that the last applied rule in $\Der$ is $\mbox{\bf(\tyrule-sub)}$, that is,
$\Der$ is of the following form:
\[%
\begin%
{array}{c}
\infer[\mbox{\bf(\tyrule-sub)}]
      {\Sigma; \vB; \Delta \tpjg e:T}
      {\Der_1::\Sigma; \vB; \Delta \tpjg e:T' & \Sigma;\vB \temd T'\tyleq T}
\end{array}\]

Let us prove (1) by induction on $\height{\Der}$.
By induction hypothesis on $\Der_1$, there exists a derivation $\Der'_1::\Sigma; \vB; \Delta
\tpjg e:T'$ such that $\height{\Der'_1}\leq\height{\Der_1}$ and
the last applied rule in $\Der'_1$ is $\mbox{\bf(\tyrule-tup)}$:
\[%
\begin%
{array}{c}
\infer[\mbox{\bf(\tyrule-tup)}]
      {\Sigma;\vB;\Delta\tpjg\dtuple{e_1,e_2}:T'_1\tand T'_2}
      {\Der'_{21}::\Sigma; \vB; \Delta\tpjg e_1: T'_1 &
       \Der'_{22}::\Sigma; \vB; \Delta\tpjg e_2: T'_2} \\[2pt]
\end{array}\]
where $T'=T'_1\tand T'_2$ and $e=\dtuple{e_1,e_2}$.
By one of the regularity condition,
$T=T_1\tand T_2$ for some $T_1$ and $T_2$.
By another regularity condition,
both
$\Sigma;\vB\temd T'_1\tyleq T_1$
and
$\Sigma;\vB\temd T'_2\tyleq T_2$ hold.
By applying
$\mbox{\bf(\tyrule-sub)}$ to $\Der'_{21}$, one obtains
$\Der_{21}::\Sigma; \vB; \Delta\tpjg e_1: T_1$.
By applying
$\mbox{\bf(\tyrule-sub)}$ to $\Der'_{22}$, one obtains
$\Der_{22}::\Sigma; \vB; \Delta\tpjg e_2: T_2$.
Let $\Der'$ be
\[%
\begin%
{array}{c}
\infer[\mbox{\bf(\tyrule-tup)}]
      {\Sigma;\vB;\Delta\tpjg\dtuple{e_1,e_2}:T_1\tand T_2}
      {\Der_{21}::\Sigma; \vB; \Delta\tpjg e_1: T_1 &
       \Der_{22}::\Sigma; \vB; \Delta\tpjg e_2: T_2} \\[2pt]
\end{array}\]
and the proof for (1) is done since
$\height{\Der'}=1+\max(\height{\Der_{21}}, \height{\Der_{22}})$,
which equals
$1+1+\max(\height{\Der'_{21}}, \height{\Der'_{22}})=1+\height{\Der'_1}\leq 1+\height{\Der_1}=\height{\Der}$.

Let us prove (2) by induction on $\height{\Der}$.
By induction hypothesis on $\Der_1$, there exists a derivation $\Der'_1::\Sigma; \vB; \Delta
\tpjg e:T'$ such that $\height{\Der'_1}\leq\height{\Der_1}$ and
the last applied rule in $\Der'_1$ is $\mbox{\bf(\tyrule-lam)}$:
\[%
\begin%
{array}{c}
\infer[\mbox{\bf(\tyrule-lam)}]
      {\Sigma;\vB;\Delta\tpjg\dlam{x}{e_1}:T'_1\timp T'_2}
      {\Der'_2::\Sigma; \vB; \Delta,x:T'_1 \tpjg e_1:T'_2} \\[2pt]
\end{array}\]
where $T'=T'_1\timp T'_2$ and $e=\dlam{x}{e_1}$.
By one of the regularity conditions,
$T=T_1\timp T_2$ for some $T_1$ and $T_2$.
By another regularity condiditon,
both
$\Sigma;\vB\temd T_1\tyleq T'_1$
and
$\Sigma;\vB\temd T'_2\tyleq T_2$ hold.
Hence, by
Lemma~\ref{lemma:var_eq},
there is a derivation $\Der''_2 ::\Sigma; \vB;
\Delta,x:T_1 \tpjg e_1:T'_2$ such that
$\height{\Der''_2}=\height{\Der'_2}$. Let $\Der'$ be the following
derivation,
\[\begin{array}{c}
\infer[\mbox{\bf(\tyrule-lam)}]
      {\Sigma;\vB;\Delta\tpjg \dlam{x}{e_1}: T_1\timp T_2}
      {\infer[\mbox{\bf(\tyrule-sub)}]
             {\Sigma;\vB;\Delta, x:T_1\tpjg e_1:T_2}
             {\Der''_2::\Sigma;\vB;\Delta, x:T_1\tpjg e_1:T'_2 &
              \Sigma;\vB\temd T'_2\tyleq T_2}}
\end{array}\]
and the proof for (2) is done since
$\height{\Der'}=1+1+\height{\Der''_2}=1+1+\height{\Der'_2}=1+\height{\Der'_1}\leq 1+\height{\Der_1}=\height{\Der}$.

The rest of statements (3), (4), (5), and (6) can all be proven similarly.
\hfill
\end{proof}

\begin%
{theorem}%
[Subject Reduction in $\ATSzero$]
\label{theorem:subject_reduction_in_ATS0}
Assume $\Der::\Sigma;\vB;\Delta\tpjg e:T$ in $\ATSzero$ and $e\eval e'$
holds. Then $\Sigma;\vB;\Delta\tpjg e':T$ is also derivable in $\ATSzero$.
\end{theorem}
\begin{proof}
The proof proceeds by induction on $\height{\Der}$.
\begin{itemize}
\item
The last applied rule in $\Der$ is $\mbox{\bf(\tyrule-sub)}$:
\[%
\begin%
{array}{c}
\infer[]
      {\Sigma;\vB;\Delta\tpjg e:T}
      {\Der_1::\Sigma;\vB;\Delta\tpjg e:T' & \Sigma\temd T'\tyleq T}
\end{array}\]
By induction hypothesis on $\Der_1$, $\Der'_1::\Sigma;\vB;\Delta\tpjg
e':T'$ is derivable, and thus the following derivation is obtained:
\[\begin{array}{c}
\infer[]
      {\Sigma;\vB;\Delta\tpjg e':T}
      {\Der'_1::\Sigma;\vB;\Delta\tpjg e':T' & \Sigma\temd T'\tyleq T}
\end{array}\]
\item
The last applied rule in $\Der$ is not $\mbox{\bf(\tyrule-sub)}$.
Assume that $e=E[e_0]$ and $e'=E[e'_0]$, where $e_0$ is a redex and
$e'_0$ is a reduct of $e_0$. All the cases where $E$ is not $[]$ can
be readily handled, and some details are given as follows on the case
where $E=[]$ (that is, $e$ is itself a redex).
\begin{itemize}
\item
$\Der$ is of the following form:
\[%
\begin%
{array}{c}
\infer[\mbox{\bf(\tyrule-fst)}]
      {\Sigma;\vB;\Delta\tpjg\dfst(\dtuple{v_{11},v_{12}}): T_1}
      {\Der_1::\Sigma;\vB;\Delta\tpjg \dtuple{v_{11},v_{12}}:T_1\tand T_2}
\end{array}\]
where $T=T_1$ and $e=\dfst(\dtuple{v_{11},v_{12}})$.
By Lemma~\ref{lemma:inversion}, $\Der_1$ may be assumed to be
of the following form:
\[%
\begin%
{array}{c}
\infer[\mbox{\bf(\tyrule-tup)}]
      {\Sigma;\vB;\Delta\tpjg\dtuple{v_{11},v_{12}}: T_1\tand T_2}
      {\Der_{21}::\Sigma;\vB;\Delta\tpjg v_{11}: T_1 & \Der_{22}::\Sigma;\vB;\Delta\tpjg v_{12}: T_2}
\end{array}\]
Note that $e'=v_{11}$, and the case concludes.
\item
$\Der$ is of the following form:
\[\begin{array}{c}
\infer[\mbox{\bf(\tyrule-app)}]
      {\Sigma;\vB;\Delta\tpjg\dapp{\dlam{x}{e_1}}{v_2}: T_2}
      {\Der_1::\Sigma;\vB;\Delta\tpjg \dlam{x}{e_1}:T_1\timp T_2 & \Der_2::\Sigma;\vB;\Delta\tpjg v_2:T_1}
\end{array}\]
where $T=T_2$ and $e=\dapp{\dlam{x}{e_1}}{v_2}$.
By Lemma~\ref{lemma:inversion}, $\Der_1$ may be assumed to be
of the following form:
\[%
\begin%
{array}{c}
\infer[]
      {\Sigma;\vB;\Delta\tpjg\dlam{x}{e_1}: T_1\timp T_2}
      {\Sigma;\vB;\Delta,x:T_1\tpjg e_1:T_2}
\end{array}\]
By Lemma~\ref{lemma:substitution} (Substitution),
$\Sigma;\vB;\Delta\tpjg \subst{v_2}{x}{e_1}:T_2$ is derivable.  Note
that $e'=\subst{v_2}{x}{e_1}$, and the case concludes.
\end{itemize}
All of the other cases can be handled similarly.
\end{itemize}
\hfill
\end{proof}

For a less involved presentation,
let us assume that any well-typed closed value of the form
$\dcf\{\vec{s}\}(v_1,\ldots,v_n)$ is a redex, that is, the dynamic constant
function $\dcf$ is well-defined at the arguments $v_1,\ldots,v_n$.
\begin%
{theorem}%
[Progress in $\ATSzero$]
\label{theorem:progress_in_ATS0}
Assume that $\Der::\emptyset;\emptyset;\emptyset\tpjg e:T$ in $\ATSzero$.
Then either $e$ is a value or $e\eval e'$ holds for some dynamic term $e'$.
\end{theorem}
\begin{proof}
With Lemma~\ref{lemma:canonical_forms} (Canonical Forms), the proof
proceeds by a straightforward structural induction on $\Der$.  \hfill
\end{proof}

By Theorem~\ref{theorem:subject_reduction_in_ATS0} and
Theorem~\ref{theorem:progress_in_ATS0}, it is clear that for each
closed well-typed dynamic term $e$, $e\meval v$ holds for some value
$v$, or there is an infinite reduction sequence starting from $e$:
$e=e_0\eval e_1\eval e_2\eval\cdots$. In other words, the evaluation
of a well-typed program in $\ATSzero$ either reaches a value or goes
on forever (as it can never get stuck). This meta-property of
$\ATSzero$ is often referred to as its type-soundness. Per Robin Milner, a
catchy slogan for type-soundness states that {\em a well-typed program
  can never go wrong}.

\def\tyerase#1{\parallel\kern-1.75pt#1\kern-1.5pt\parallel}
\def\pferase#1{|#1|}
\def\pferasep#1{\pferase{#1}_p}
\def\pferaset#1{\pferase{#1}_t}
\begin%
{figure}[thp]
\[%
\begin%
{array}{rcl}
\tyerase{x} & = & x \\
\tyerase{\dcx\{\vec{s}\}(e_1,\ldots,e_n)} & = & \dcx(\tyerase{e_1},\ldots,\tyerase{e_n}) \\
\tyerase{\dlam{x}{e}} & = & \dlam{x}{\tyerase{e}} \\
\tyerase{\dapp{e_1}{e_2}} & = & \dapp{\tyerase{e_1}}{\tyerase{e_2}} \\
\tyerase{\;\diguard{e}} & = & \tyerase{e} \\
\tyerase{\;\deguard{e}} & = & \tyerase{e} \\
\tyerase{\dassert{e}} & = & \tyerase{e} \\
\tyerase{\dsletin{\dassert{x}=e_1}{e_2}} & = & \dletin{x=\tyerase{e_1}}{\tyerase{e_2}} \\
\tyerase{\dslam{a}{e}} & = & \tyerase{e} \\
\tyerase{\dsapp{e}{s}} & = & \tyerase{e} \\
\end{array}
\]
\caption%
{The type-erasure function $\tyerase{\cdot}$ on dynamic terms}
\label{figure:ATS0_type_erasure}
\end{figure}
After a program in $\atslang$ passes type-checking, it goes through a
process referred to as type-erasure to have the static terms inside it
completely erased. In Figure~\ref{figure:ATS0_type_erasure}, a
function performing type-erasure is defined, which maps each dynamic
term in $\ATSzero$ to an untyped dynamic term in $\lamdyn$.

In order to guarantee that a value in $\ATSzero$ is mapped to
another value in $\lamdyn$ by the function $\tyerase{\cdot}$, the following
syntactic restriction is needed:
\begin%
{itemize}
\item
Only when $e$ is a value can the dynamic term $\diguard{e}$ be formed.
\item
Only when $e$ is a value can the dynamic term $\dslam{a}{e}$ be formed.
\end{itemize}
This kind of restriction is often referred to as value-form restriction.

\begin%
{proposition}\label{prop:type-erasure-0}
With the value-form restriction being imposed, $\tyerase{v}$ is a value in
$\lamdyn$ for every value $v$ in $\ATSzero$.
\end{proposition}
\begin%
{proof}
By structural induction on $v$.
\hfill
\end{proof}
Note that it is certainly possible to have a non-value $e$ in
$\ATSzero$ whose type-erasure is a value in $\lamdyn$.  From this
point on, the value-form restriction is always assumed to have been
imposed when type-erasure is performed.

\begin%
{proposition}\label{prop:type-erasure-1}
Assume that $e_1$ is a well-typed closed dynamic term in $\ATSzero$.
If $e_1\eval e_2$ holds, then either $\tyerase{e_1}=\tyerase{e_2}$ or
$\tyerase{e_1}\eval\tyerase{e_2}$ holds in $\lamdyn$.
\end{proposition}
\begin%
{proof}
By a careful inspection of the forms of redexes in Definition~\ref{ATS0_redex_def}.
\hfill
\end{proof}

\begin%
{proposition}\label{prop:type-erasure-2}
Assume that $e_1$ is a well-typed closed dynamic term in $\ATSzero$.
If $\tyerase{e_1}\eval e'_2$ holds in $\lamdyn$, then there exists
$e_2$ such that $e_1\meval e_2$ holds in $\ATSzero$ and
$\tyerase{e_2}=e'_2$.
\end{proposition}
\begin%
{proof}
By induction on the height of the typing derivation for $e_1$.
\hfill
\end{proof}

By Proposition~\ref{prop:type-erasure-1} and
Proposition~\ref{prop:type-erasure-2}, it is clear that type-erasure
cannot alter the dynamic semantics of a well-typed dynamic term in
$\ATSzero$.

The formulation of $\ATSzero$ presented in this section is of a
minimalist style. In particular, the constraint relation in $\ATSzero$
is treated abstractly. In practice, if a concrete instance of
$\ATSzero$ is to be implemented, then rules need to be provided for
simplifying constraints. For instance, the following rule may be present:
\[
\infer
{\Sigma;\vB\temd\tint(I_1)\tyleq\tint(I_2)}{\Sigma;\vB\temd I_1=I_2}
\]
With this rule, $\tint(I_1)\tyleq\tint(I_2)$ can be simplified to
the constraint $I_1=I_2$, where the equality is on static integer terms.
The following rule may also be present:
\[
\infer
{\Sigma;\vB\temd\tlist(T_1,I_1)\tyleq\tlist(T_2,I_2)}
{\Sigma;\vB\temd T_1\tyleq T_2 & \Sigma;\vB\temd I_1=I_2}
\]
With this rule, $\tlist(T_1,I_1)\tyleq\tlist(T_2,I_2)$ can be
simplified to the two constraints $T_1\tyleq T_2$ and $I_1=I_2$.

For those interested in implementing an applied type system, please
find more details in a paper on $\DML$~\cite{DML-jfp07}, which is
regarded a special kind of applied type system.

\section%
{Formal Development of $\ATSproof$}\label{section:ATSproof}
Let us extend $\ATSzero$ to $\ATSproof$ in this section with support
for programming with theorem-proving (PwTP).

A great limitation on employing $\ATSzero$ as the basis for a
practical programming language lies in the very rigid handling of
constraint-solving in $\ATSzero$. One is often forced to impose
various {\em ad hoc} restrictions on the syntactic form of a
constraint that can actually be supported in practice (so as to match
the capability of the underlying constraint-solver), greatly
diminishing the effectiveness of using types to capture programming
invariants. For instance, only quantifier-free constraints that can be
translated into problems of linear integer programming are allowed in
the DML programming language~\cite{XiWebDML}.

With PwTP being supported in a programming language, programming and
theorem-proving can be combined in a syntactically intertwined
manner~\cite{CPwTP-icfp05}; if a constraint cannot be handled directly
by the underlying constraint-solver, then it is possible to simplify
the constraint or even eliminate it through explicit proof
construction. PwTP advocates an open style of constraint-solving by
providing a means within the programming language itself to allow the
programmer to actively participate in constraint-solving. In other
words, PwTP can be viewed as a programming paradigm for internalizing
constraint-solving.

\def\TP{T^{*}}
\def\prleq{\leq_{pr}}
\begin{figure}[thp]
\[%
\begin%
{array}{ccl}
\prleq %
      & : & (\sprop,\sprop) \Simp \sbool \\
\tand & : & (\sprop,\sprop) \Simp \sprop \\
\tand & : & (\sprop,\stype) \Simp \stype \\
\timp & : & (\sprop,\sprop) \Simp \sprop \\
\timp & : & (\sprop,\stype) \Simp \stype \\
\Band & : & (\sbool,\sprop) \Simp \sprop \\
\Bimp & : & (\sbool,\sprop) \Simp \sprop \\
\forall_\sigma & : & (\sigma\simp\sprop)\Simp\sprop \\
\exists_\sigma & : & (\sigma\simp\sprop)\Simp\sprop \\
\end{array}\]
\caption{Additional static constants in $\ATSproof$}
\label{figure:ATSproof_sconsts}
\end{figure}
Let us now start with the formulation of $\ATSproof$, which extends
that of $\ATSzero$ fairly lightly.  In addition to the base sorts in
$\ATSzero$, $\ATSproof$ contains another base sort $\sprop$, which is
for static terms representing types for proofs. A static term of the
sort $\sprop$ may be referred to as a prop (or, sometimes, a type for
proofs). Also, it is assumed that the static constants listed in
Figure~\ref{figure:ATSproof_sconsts} are included in $\ATSproof$.
Note that the symbols referring to these static constants may be
overloaded. In the following representation, $P$ stands for a prop,
$T$ stands for a type, and $\TP$ stands for either a prop or a type.

\def\dtuplepp#1{\langle#1\rangle_{\it pp}}
\def\dtuplept#1{\langle#1\rangle_{\it pt}}
\def\dtuplett#1{\langle#1\rangle_{\it tt}}
\def\dlampp#1#2{{\bf lam}_{\it pp}\;#1.\;#2}
\def\dlampt#1#2{{\bf lam}_{\it pt}\;#1.\;#2}
\def\dlamtt#1#2{{\bf lam}_{\it tt}\;#1.\;#2}
\def\dapppp#1#2{{\bf app}_{\it pp}(#1,#2)}
\def\dapptp#1#2{{\bf app}_{\it tp}(#1,#2)}
\def\dapptt#1#2{{\bf app}_{\it tt}(#1,#2)}
The syntax for dynamic terms in $\ATSproof$ is essentially the same as
that in $\ATSzero$ but with a few minor changes to be mentioned as
follows.  Some dynamic constructs in $\ATSzero$ need to be split when
they are incorporated into $\ATSproof$.  The construct
$\dtuple{e_1,e_2}$ for forming tuples is split into
$\dtuplepp{e_1,e_2}$, $\dtuplept{e_1,e_2}$, and $\dtuplett{e_1,e_2}$
for prop-type pairs, prop-type pairs and type-type pairs,
respectively. For instance, a prop-type pair is one where the first
component is assigned a prop and the second one a type. Note that
there are no type-prop pairs.  The construct $\dlam{x}{e}$ for forming
lambda-abstractions is split into $\dlampp{x}{e}$, $\dlampt{x}{e}$,
and $\dlamtt{x}{e}$ for prop-prop functions, prop-type functions and
type-type functions, respectively. For instance, a prop-type function
is one where the argument is assigned a prop and the body a type. The
construct $\dapp{e_1}{e_2}$ for forming applications is split into
$\dapppp{e_1}{e_2}$, $\dapptp{e_1}{e_2}$, and $\dapptt{e_1}{e_2}$ for
prop-prop applications, type-prop applications and type-type
applications. For instance, a type-prop application is one where the
function part is assigned a type and the argument a prop.
Note that there are no type-prop functions.

The dynamic variable contexts in $\ATSproof$ are defined as follows:
\[%
\begin%
{array}{lrcl}
\mbox{dynamic var. ctx.} & \Delta & ::= & \emptyset \mid \Delta, x:\TP \\
\end{array}
\]

The regularity conditions on $\tyleq$ needs to be extended with the
following two for the new forms of types:
\begin%
{itemize}
\item[3.2]
$\Sigma;\vB\temd P_1\tand T_2\tyleq P'_1\tand T'_2$ implies
$\Sigma;\vB\temd P_1\prleq P'_1$ and $\Sigma;\vB\temd T_2\tyleq T'_2$.
\item[4.2]
$\Sigma;\vB\temd P_1\timp T_2\tyleq P'_1\timp T'_2$ implies
$\Sigma;\vB\temd P'_1\prleq P_1$ and $\Sigma;\vB\temd T_2\tyleq T'_2$.
\end{itemize}
It should be noted that there are no regularity conditions imposed
on props (as it is not expected for proofs to have any computational
meaning).

There are two kinds of typing rules in $\ATSproof$: p-typing rules and
t-typing rules, where the former is for assigning props to dynamic
terms (encoding proofs) and the latter for assigning types to dynamic
terms (to be evaluated).  The typing rules for $\ATSproof$ are
essentially those for $\ATSzero$ listed in
Figure~\ref{figure:typing_rules_for_ATS0_dynamics} except for the
following changes:
\begin%
{itemize}
\item
Each occurrence of $T$ in the rules for $\ATSzero$ needs to be 
replaced with $\TP$.
\item
The premisses of each p-typing rule (that is, one for assigning a prop
to a dynamic term) are required to be p-typing rules themselves.
\end{itemize}
As an example, let us take a look at the following rule:
\[%
\begin%
{array}{c}
\infer[\mbox{\bf(\tyrule-sub)}]
      {\Sigma; \vB; \Delta \tpjg e:T'}
      {\Sigma; \vB; \Delta \tpjg e:T & \Sigma;\vB \temd T\tyleq T'} \\[2pt]
\end{array}\]
which yields the following two valid versions:
\[%
\begin%
{array}{c}
\infer[\mbox{\bf(\tyrule-sub-p)}]
      {\Sigma; \vB; \Delta \tpjg e:P'}
      {\Sigma; \vB; \Delta \tpjg e:P & \Sigma;\vB \temd P\prleq P'} \\[2pt]
\infer[\mbox{\bf(\tyrule-sub-t)}]
      {\Sigma; \vB; \Delta \tpjg e:T'}
      {\Sigma; \vB; \Delta \tpjg e:T & \Sigma;\vB \temd T\tyleq T'} \\[2pt]
\end{array}\]
As another example, let us take a look at the following rule:
\[%
\begin%
{array}{c}
\infer[\mbox{\bf(\tyrule-fst)}]
      {\Sigma; \vB; \Delta\tpjg \dfst(e):T_1}
      {\Sigma; \vB; \Delta\tpjg e: T_1 * T_2}
\end{array}\]
which yields the following two valid versions:
\[%
\begin%
{array}{c}
\infer[\mbox{\bf(\tyrule-fst-pp)}]
      {\Sigma; \vB; \Delta\tpjg \dfst(e):P_1}
      {\Sigma; \vB; \Delta\tpjg e: P_1 * P_2} \\[2pt]
\infer[\mbox{\bf(\tyrule-fst-tt)}]
      {\Sigma; \vB; \Delta\tpjg \dfst(e):T_1}
      {\Sigma; \vB; \Delta\tpjg e: T_1 * T_2} \\[2pt]
\end{array}\]
Note that there is no type of the form $T_1 \tand P_2$ (for the sake
of simplicity). The following version is invalid:
\[%
\begin%
{array}{c}
\infer[\mbox{\bf(\tyrule-fst-pt)}]
      {\Sigma; \vB; \Delta\tpjg \dfst(e):P_1}{\Sigma; \vB; \Delta\tpjg e: P_1 * T_2}
\end{array}
\]
because a p-typing rule cannot have any t-typing rule as its premise.
Instead, the following typing rule is introduced as the elimination
rule for $P_1*T_2$:
\[%
\begin%
{array}{c}
\infer[\mbox{\bf(\tyrule-$\tand$-elim-pt)}]
      {\Sigma; \vB; \Delta\tpjg \dletin{\dtuplept{x_1,x_2}=e}{e_0}:T_0}
      {\Sigma; \vB; \Delta\tpjg e: P_1 * T_2 & \Sigma;\vB;\Delta,x_1:P_1,x_2:T_2 \tpjg e_0:T_0}
\end{array}
\]
As yet another example, let us take a look at the following rule:
\[%
\begin%
{array}{c}
\infer[\mbox{\bf(\tyrule-app)}]
      {\Sigma; \vB; \Delta\tpjg \dapp{e_1}{e_2}: \TP_2}
      {\Sigma; \vB; \Delta\tpjg e_1: \TP_1\timp \TP_2 & \Sigma; \vB; \Delta\tpjg e_2: \TP_1} \\[2pt]
\end{array}\]
which yields the following three versions:
\[%
\begin%
{array}{c}
\infer[\mbox{\bf(\tyrule-app-pp)}]
      {\Sigma; \vB; \Delta\tpjg \dapppp{e_1}{e_2}: P_2}
      {\Sigma; \vB; \Delta\tpjg e_1: P_1\timp P_2 & \Sigma; \vB; \Delta\tpjg e_2: P_1} \\[2pt]
\infer[\mbox{\bf(\tyrule-app-tp)}]
      {\Sigma; \vB; \Delta\tpjg \dapptp{e_1}{e_2}: T_2}
      {\Sigma; \vB; \Delta\tpjg e_1: P_1\timp T_2 & \Sigma; \vB; \Delta\tpjg e_2: P_1} \\[2pt]
\infer[\mbox{\bf(\tyrule-app-tt)}]
      {\Sigma; \vB; \Delta\tpjg \dapptt{e_1}{e_2}: T_2}
      {\Sigma; \vB; \Delta\tpjg e_1: T_1\timp T_2 & \Sigma; \vB; \Delta\tpjg e_2: T_1} \\[2pt]
\end{array}\]
The first one is a p-typing rule while the other two are t-typing rules.

In $\ATSproof$, the two sorts $\sbool$ and $\sprop$ are intimately
related but are also fundamentally different.  Gaining a solid
understanding of the relation between these two is the key to
understanding the design of $\ATSproof$.  One may see $\sprop$ as an
internalized version of $\sbool$. Given a static boolean term $B$, its
truth value is determined by a constraint-solver outside
$\ATSproof$. Given a static term $P$ of the sort $\sprop$, a proof of
$P$ can be constructed inside $\ATSproof$ to attest to the validity of
the boolean term encoded by $P$. For clarification, let us see a
simple example illustrating the relation between $\sbool$ and $\sprop$
in concrete terms.

\def\punit{\tunit}
\def\fact{{\it fact}}
\def\factp{{\it fact\_p}}
\def\factpbas{\mbox{\tt fact\_p\_bas}}
\def\factpind{\mbox{\tt fact\_p\_ind}}
\def\factb{{\it fact\_b}}
\def\factbbas{\mbox{\tt fact\_b\_bas}}
\def\factbind{\mbox{\tt fact\_b\_ind}}
\def\ffactp{{\it f\_fact\_p}}
\def\ffactb{{\it f\_fact\_b}}
\begin%
{figure}[thp]
\begin%
{minipage}{206pt}
\begin%
{verbatim}
dataprop
fact_p(int, int) =
  | fact_p_bas(0, 1) of ()
  | {n:nat}{r:int}
    fact_p_ind(n+1, (n+1)*r) of fact_p(n, r)
\end{verbatim}
\end{minipage}
\vspace{12pt}
\caption{A dataprop for encoding the factorial function}
\label{figure:dataprop_fact_p}
\end{figure}
\begin%
{figure}[thp]
\begin%
{minipage}{276pt}
\begin%
{verbatim}
stacst
fact_b : (int, int) -> bool
praxi
fact_b_bas
(
  // argless
) : [fact_b(0, 1)] unit_p
praxi
fact_b_ind{n:int}{r:int}
(
  // argless
) : [n >= 0 && fact_b(n, r) ->> fact_b(n+1, (n+1)*r)] unit_p
\end{verbatim}
\end{minipage}
\vspace{12pt}
\caption{A static predicate and two associated proof functions}
\label{figure:stacst_fact_b}
\end{figure}
In Figure~\ref{figure:dataprop_fact_p}, the dataprop $\factp$
declared in $\atslang$ is associated with two proof constructors
that are assigned the following c-types (or, more precisely, c-props):
\[
\begin%
{array}{ccl}
\factpbas &~~:~~& \factp(0, 1) \\
\factpind &~~:~~& \forall n:\snat.\forall r:\sint.~(\factp(n, r))\Timp\factp(n+1, (n+1)*r) \\
\end{array}
\]
Let $\fact(n)$ be the value of the factorial function on $n$,
where $n$ ranges over natural numbers.  Given a natural number $n$ and
an integer $r$, the prop $\factp(n, r)$ encodes the relation $\fact(n)=r$.
In other words, if a proof of the prop $\factp(n, r)$ can be constructed,
then $\fact(n)$ equals $r$.

In Figure~\ref{figure:stacst_fact_b}, a static predicate $\factb$ is
introduced, which corresponds to $\factp$.  Given a natural number $n$
and an integer $r$, $\factb(n, r)$ simply means $\fact(n)=r$. The two
proof functions $\factbbas$ and $\factbind$ are assigned the following
c-props:
\[
\begin%
{array}{ccl}
\factbbas &\kern-6pt:\kern-6pt& ()\Timp\factb(0, 1)\Band\punit \\
\factbind &\kern-6pt:\kern-6pt& \forall n:\sint.\forall r:\sint.~()\Timp(n\geq 0\land\factb(n, r)\limplies\factb(n+1, (n+1)\cdot r))\Band\punit\\
\end{array}
\]
where $\punit$ is the unit prop (instead of the unit type)
that encodes the static truth value $\ttrue$.
Note that the keyword $\mbox{\it praxi}$ in $\atslang$ is used to
introduce proof functions that are treated as axioms.

\begin%
{figure}[thp]
\begin%
{minipage}{252pt}
\begin%
{verbatim}
fun
f_fact_p
  {n:nat}
(
  n: int(n)
) : [r:int]
  (fact_p(n, r) | int(r)) = let
//
fun
loop
{ i:nat
| i <= n
} {r:int}
(
  pf: fact_p(i, r)
| i: int(i), r: int(r)
) : [r:int] (fact_p(n, r) | int(r)) =
  if i < n then
    loop(fact_p_ind(pf) | i+1, (i+1)*r) else (pf | r)
  // end of [if]
//
in
  loop(fact_p_bas() | 0(*i*), 1(*r*))
end // end of [f_fact_p]
\end{verbatim}
\end{minipage}
\vspace{12pt}
\caption{A verified implementation of the factorial function}
\label{figure:f_fact_p}
\end{figure}
\begin%
{figure}[thp]
\begin%
{minipage}{200pt}
\begin%
{verbatim}
fun
f_fact_b
  {n:nat}
(
  n: int(n)
) : [r:int]
  (fact_b(n, r) && int(r)) = let
//
prval() = $solver_assert(fact_b_bas)
prval() = $solver_assert(fact_b_ind)
//
fun
loop
{ i:nat | i <= n}
{ r:int | fact_b(i, r) }
(
  i: int(i), r: int(r)
) : [r:int]
  (fact_b(n, r) && int(r)) =
  if i < n then loop(i+1, (i+1)*r) else (r)
//
in
  loop(0, 1)
end // end of [f_fact_b]
\end{verbatim}
\end{minipage}
\vspace{12pt}
\caption{Another verified implementation of the factorial function}
\label{figure:f_fact_b}
\end{figure}

In Figure~\ref{figure:f_fact_p}, a verified implementation of the
factorial function is given in $\atslang$. Given a natural numbers
$n$, $\ffactp$ returns an integer $r$ paired with a proof of
$\factp(n, r$) that attests to the validity of $\fact(n)=r$. Note that
this implementation makes explicit use of proofs.  The constraints
generated from type-checking the code in Figure~\ref{figure:f_fact_p}
are quantifier-free, and they can be readily solved by the built-in
constraint-solver (based on linear integer programming) for
$\atslang$.

In Figure~\ref{figure:f_fact_b}, another verified implementation of
the factorial function is given in $\atslang$. Given a natural
numbers, $\ffactb$ returns an integer $r$ plus the assertion
$\factb(n, r)$ that states $\fact(n)=r$. This implementation does not
make explicit use of proofs. Applying the keyword
$\mbox{\it\$solver\_assert}$ to a proof turns the prop of the proof
into a static boolean term (of the same meaning) and then adds the
term as an assumption to be used for solving the constraints generated
subsequently in the same scope. For instance, the two applications of
$\mbox{\it\$solver\_assert}$ essentially add the following two
assumptions:
\[
\begin%
{array}{l}
\factb(0, 1) \\
\forall n:\sint.\forall r:\sint.~n\geq 0\land\factb(n, r)\limplies\factb(n+1, (n+1)\cdot r) \\
\end{array}
\]
Note that the second assumption is universally quantified. In general,
solving constraints involving quantifiers is much more difficult than
those that are quantifier-free. For instance, the constraints generated
from type-checking the code in Figure~\ref{figure:f_fact_b} cannot be
solved by the built-in constraint-solver for $\atslang$. Instead,
these constraints need to be exported so that external
constraint-solvers (for instance, one based on the Z3
theorem-prover~\cite{Z3-tacas08}) can be invoked to solve them.

By comparing these two verified implementations of the factorial
function, one sees a concrete case where PwTP (as is done in
Figure~\ref{figure:f_fact_p}) is employed to simplify the constraints
generated from type-checking. This kind of constraint simplification
through PwTP is a form of internalization of constraint-solving, and
it can often play a pivotal r{\^o}le in practice, especially, when
there is no effective method available for solving general
unsimplified constraints.

Instead of assigning (call-by-value) dynamic semantics to the dynamic
terms in $\ATSproof$ directly, a translation often referred to as
proof-erasure is to be defined that turns each dynamic term in
$\ATSproof$ into one in $\ATSzero$ of the same dynamic semantics.

Given a sort $\sigma$, its proof-erasure $\pferase{\sigma}$ is the one
in which every occurrence of $\sprop$ in $\sigma$ is replaced with
$\sbool$.

Given a static variable context $\Sigma$, its proof-erasure
$\pferase{\Sigma}$ is obtained from replacing each declaration
$a:\sigma$ with $a:\pferase{\sigma}$.

For every static constant $\scx$ of the c-sort
$(\sigma_1,\ldots,\sigma_n)\Simp\sigma$, it is assumed that there
exists a corresponding $\scx'$ of the c-sort
$(\pferase{\sigma_1},\ldots,\pferase{\sigma_n})\Simp\pferase{\sigma}$;
this corresponding $\scx'$ may be denoted by $\pferase{\scx}$.
Note that it is possible to have $\pferase{\scx_1}=\pferase{\scx_2}$
for different constants $\scx_1$ and $\scx_2$.

Let us assume the existence of the following static constants:
\[%
\begin%
{array}{ccl}
\land & : & (\sbool, \sbool) \Simp \sbool \\
\limplies & : & (\sbool, \sbool) \Simp \sbool \\
\forall_{\sigma} & : & (\sigma\simp\sbool) \Simp \sbool \\
\exists_{\sigma} & : & (\sigma\simp\sbool) \Simp \sbool \\
\end{array}\]
Note that the symbols referring to these static constants are all
overloaded. Naturally, $\land$ and $\limplies$ are interpreted as the
boolean conjunction and boolean implication, respectively, and
$\forall_{\sigma}$ and $\exists_{\sigma}$ are interpreted as the
standard universal quantification and existential quantification,
respectively.  For instance, some pairs of corresponding static
constants are listed as follows:
\begin%
{itemize}
\item
The boolean implication function $\limplies$
corresponds to the prop predicate $\prleq$.
\item
The boolean implication function $\limplies$
corresponds to the prop constructor
$\timp$ of the c-sort $(\sprop,\sprop)\Simp\sprop$.
\item
The boolean implication function $\limplies$
corresponds to the prop constructor
$\Bimp$ of the c-sort $(\sbool,\sprop)\Simp\sprop$.
\item
The boolean conjunction function $\land$
corresponds to the prop constructor
$~\tand~$ of the c-sort $(\sprop,\sprop)\Simp\sprop$.
\item
The boolean conjunction function $\land$
corresponds to the prop constructor
$\Band$ of the c-sort $(\sbool,\sprop)\Simp\sprop$.
\item
The type constructor $\Band$ of the c-sort
$(\sbool,\stype)\Simp\stype$ corresponds to the type constructor
$\tand$ of the c-sort $(\sprop,\stype)\Simp\stype$.
\item
The type constructor $\Bimp$ of the c-sort $(\sbool,\stype)\Simp\stype$
corresponds to the type constructor $\timp$ of the c-sort
$(\sprop,\stype)\Simp\stype$.
\item
For each sort $\sigma$,
the universal quantifier
$\forall_{\sigma}$ of the sort $(\sigma\simp\sbool)\Simp\sbool$
corresponds to
the universal quantifier
$\forall_{\sigma}$ of the sort $(\sigma\simp\sprop)\Simp\sprop$.
\item
For each sort $\sigma$,
the existential quantifier
$\exists_{\sigma}$ of the sort $(\sigma\simp\sbool)\Simp\sbool$
corresponds to
the existential quantifier
$\exists_{\sigma}$ of the sort $(\sigma\simp\sprop)\Simp\sprop$.
\end{itemize}

For every static term $s$, $\pferase{s}$ is the static term obtained
from replacing in $s$ each $\sigma$ with $\pferase{\sigma}$ and each $\scx$
with $\pferase{\scx}$.
\begin%
{proposition}
Assume that $\Sigma\tpjg s:\sigma$ is derivable.
Then $\pferase{\Sigma}\tpjg \pferase{s}:\pferase{\sigma}$ is also derivable.
\end{proposition}
\begin%
{proof}
By induction on the sorting derivation of $\Sigma\tpjg s:\sigma$.
\hfill
\end{proof}

For a sequence $\vB$ of static boolean terms,
$\pferase{\vB}$ is the sequence obtained from applying $\pferase{\cdot}$
to each $B$ in $\vB$.

There are two functions $\pferasep{\cdot}$ and $\pferaset{\cdot}$ for
mapping a given dynamic variable context $\Delta$ to a sequence of
boolean terms and a dynamic variable context, respectively:
\begin%
{itemize}
\item
$\pferasep{\Delta}$ is a sequence of boolean terms $\vB$
such that each $B$ in $\vB$ is $\pferase{P}$ for some $a:P$ declared
in $\Sigma$.
\item
$\pferaset{\Delta}$ is a dynamic variable context such each
declaration in it is of the form $a:\pferase{T}$ for some $a:T$ declared
in $\Sigma$.
\end{itemize}

\begin%
{figure}
\[%
\begin%
{array}{rcl}
\pferase{x} & = & x \\
\pferase{\dcx\{\vec{s}\}(\vec{e})}
                        & = & \dcx\{\pferase{\vec{s}}\}(\pferase{\vec{e}}) \\
\pferase{\dtuplept{e_1,e_2}} %
                        & = & \dassert{\pferase{e_2}} \\
\pferase{\dtuplett{e_1,e_2}} %
                        & = & \dtuplett{\pferase{e_1},\pferase{e_2}} \\
\pferase{\dfst(e)}      & = & \dfst(\pferase{e}) \\
\pferase{\dsnd(e)}      & = & \dsnd(\pferase{e}) \\
\pferase{\dletin{\dtuplept{x_p,x_t}=e_1}{e_2}} %
                        & = & \dletin{\dassert{x_t}=\;\pferase{e_1}}{\pferase{e_2}} \\
\pferase{\dlampt{x}{e}} & = & \diguard{\pferase{e}} \\
\pferase{\dlamtt{x}{e}} & = & \dlam{x}{\!\!\pferase{e}} \\
\pferase{\dapptp{e_1}{e_2}} & = & \deguard{\pferase{e_1}} \\
\pferase{\dapptt{e_1}{e_2}} & = & \dapp{\pferase{e_1}}{\pferase{e_2}} \\
\pferase{\;\diguard{e}} & = & \diguard{\pferase{e}} \\
\pferase{\;\deguard{e}} & = & \deguard{\pferase{e}} \\
\pferase{\dassert{e}}   & = & \dassert{\pferase{e}} \\
\pferase{\dsletin{\dassert{x}=e_1}{e_2}} %
                       & = & \dsletin{\dassert{x}=\;\pferase{e_1}}{\pferase{e_2}} \\
\pferase{\dslam{a}{e}}  & = & \dslam{a}{\pferase{e}} \\
\pferase{\dsapp{e}{s}}  & = & \dsapp{\pferase{e}}{\pferase{s}} \\
\end{array}\]
\caption%
{The proof-erasure function $\pferase{\cdot}$ on dynamic terms}
\label{figure:ATSpf_proof_erasure}
\end{figure}
The proof-erasure function on dynamic terms is defined in
Figure~\ref{figure:ATSpf_proof_erasure}.  Clearly, given a dynamic
term $e$ in $\ATSproof$, $\pferase{e}$ is a dynamic term in $\ATSzero$
if it is defined.

\def\pneg{{\it not}}
As the proof-erasure of $\prleq$
is chosen to be the boolean implication function,
it needs to be assumed that
$\Sigma;\vB\tpjg P_1\prleq P_2$
implies
$\pferase{\Sigma};\pferase{\vB}\tpjg\pferase{P_1}\limplies\pferase{P_2}$
\begin%
{lemma}%
[Constraint Internalization]%
\label{lemma:ATSproof_prop2bool}
Assume that the typing judgment $\Sigma;\vB;\Delta\tpjg e:P$
is derivable in $\ATSproof$.  Then the constraint
$\pferase{\Sigma};\pferase{\vB},\pferasep{\Delta}\temd\pferase{P}$ holds.
\end{lemma}
\begin%
{proof}
By structural induction on the typing derivation $\Der$ of
$\Sigma;\vB;\Delta\tpjg e:P$.
Note that the typing rule $\mbox{\bf(\tyrule-sub-p)}$ is handled by the assumption
that $\Sigma;\vB\tpjg P_1\prleq P_2$ implies
$\pferase{\Sigma};\pferase{\vB}\tpjg\pferase{P_1}\limplies\pferase{P_2}$ for any props $P_1$ and $P_2$.
\begin%
{itemize}
\item
Assume that the last applied rule in $\Der$ is $\mbox{\bf(\tyrule-tup-pp)}$:
\[%
\begin%
{array}{c}
\infer[\mbox{\bf(\tyrule-tup-pp)}]
      {\Sigma;\vB;\Delta\tpjg\dtuplepp{e_1,e_2}: P_1\tand P_2}
      {\Der_{1}::\Sigma;\vB;\Delta\tpjg e_1: P_1 & \Der_{2}::\Sigma;\vB;\Delta\tpjg e_2: P_2}
\end{array}\]
where $P=P_1\tand P_2$.
By induction hypothesis on $\Der_1$,
$\pferase{\Sigma};\pferase{\vB},\pferasep{\Delta}\temd\pferase{P_1}$
holds.
By induction hypothesis on $\Der_2$,
$\pferase{\Sigma};\pferase{\vB},\pferasep{\Delta}\temd\pferase{P_2}$
holds.
Note that
$\pferase{P}=\pferase{P_1\tand P_2}=\pferase{P_1}\land\pferase{P_2}$,
where $\land$ stands for the boolean conjunction. Therefore,
$\pferase{\Sigma};\pferase{\vB},\pferasep{\Delta}\temd\pferase{P}$
holds.
\item
Assume that the last applied rule in $\Der$ is
either
$\mbox{\bf(\tyrule-fst-pp)}$
or
$\mbox{\bf(\tyrule-snd-pp)}$.
This case immediately follows from the fact that
$\pferase{P_1*P_2}=\pferase{P_1}\land\pferase{P_2}$ for any props
$P_1$ and $P_2$, where $\land$ stands for the boolean conjunction
\item
Assume that the last applied rule in $\Der$ is $\mbox{\bf(\tyrule-lam-pp)}$:
\[%
\begin%
{array}{c}
\infer[\mbox{\bf(\tyrule-lam-pp)}]
      {\Sigma;\vB;\Delta\tpjg\dlampp{x_1}{e_2}: P_1\timp P_2}
      {\Der_1::\Sigma;\vB;\Delta,x_1: P_1\tpjg e_2: P_2}
\end{array}\]
where $P=P_1\timp P_2$.
By induction hypothesis on $\Der_1$,
$\pferase{\Sigma};\pferase{\vB},\pferasep{\Delta},\pferase{P_1}\temd\pferase{P_2}$
holds.
By the regularity rule
$\mbox{\bf(\regrule-cut)}$,
$\pferase{\Sigma};\pferase{\vB},\pferasep{\Delta}\temd\pferase{P_2}$ holds
whenever
$\pferase{\Sigma};\pferase{\vB},\pferasep{\Delta}\temd\pferase{P_1}$ holds.
Therefore,
$\pferase{\Sigma};\pferase{\vB},\pferasep{\Delta}\temd\pferase{P_1}\limplies\pferase{P_2}$ holds,
where $\limplies$ stands for the boolean implication.
Note that $\pferase{P}=\pferase{P_1}\limplies\pferase{P_2}$, and this case concludes.
\item
Assume that the last applied rule in $\Der$ is $\mbox{\bf(\tyrule-app-pp)}$:
\[%
\begin%
{array}{c}
\infer[\mbox{\bf(\tyrule-app-pp)}]
      {\Sigma; \vB; \Delta\tpjg \dapppp{e_1}{e_2}: P_2}
      {\Der_1::\Sigma; \vB; \Delta\tpjg e_1: P_1\timp P_2 & \Der_2::\Sigma; \vB; \Delta\tpjg e_2: P_1}
\end{array}\]
where $P=P_2$.
By induction hypothesis on $\Der_1$,
$\pferase{\Sigma};\pferase{\vB},\pferasep{\Delta}\temd\pferase{P_1}\limplies\pferase{P_2}$ holds,
where $\limplies$ stands for the boolean implication.
By induction hypothesis on $\Der_2$, the constraint
$\pferase{\Sigma};\pferase{\vB},\pferasep{\Delta}\temd\pferase{P_1}$ holds.
Therefore, the constraint
$\pferase{\Sigma};\pferase{\vB},\pferasep{\Delta}\temd\pferase{P_2}$ also holds.

\end{itemize}
The rest of the cases can be handled similarly.
\hfill
\end{proof}
Note that a proof in $\ATSproof$ can be non-constructive as it is not
expected for the proof to have any computational meaning.  In
particular, one can extend the proof construction in $\ATSproof$ with
any kind of reasoning based on classical logic (e.g., double negation
elimination).

\def\ept{e_{12}}
\def\ctype{\mbox{\it CT}}
If a c-type $\ctype$ assigned to a dynamic (proof) constant is of the
form $\forall\Sigma.\vB\Bimp(\vec{P})\Timp P_0$, then it is assumed
that the following constraint holds in $\ATSzero$:
$$\emptyset;\emptyset\temd\forall\pferase{\Sigma}.\pferase{\vB}\Bimp(\pferase{\vec{P}}\Bimp\pferase{P_0})$$
For instance, the c-types assigned to $\factpbas$ and $\factpind$ imply the validity of the following constraints:
\[%
\begin%
{array}{l}
\emptyset;\emptyset\tpjg\factb(0, 1) \\
\emptyset;\emptyset\tpjg\forall n:\sint.\forall r:\sint.~(n\geq 0\land\factb(n, r)\limplies\factb(n+1, (n+1)\cdot r)) \\
\end{array}\]
which are encoded directly into the c-types assigned to $\factbbas$ and $\factbind$.

If a c-type $\ctype$ is of the form
$\forall\Sigma.\vB\Bimp(\vec{P},T_1,\ldots,T_n)\Timp T_0$,
then $\pferase{\ctype}$ is defined as follows:
$$\forall\pferase{\Sigma}.\pferase{\vB}\Bimp(\pferase{\vec{P}}\Bimp((\pferase{T_1},\ldots,\pferase{T_n})\Timp\pferase{T_0}))$$
If a dynamic constant
$\dcx$ is assigned the c-type $\ctype$ in $\ATSproof$,
then it is assumed to be of the c-type $\pferase{\ctype}$ in $\ATSzero$.
\begin%
{theorem}%
\label{theorem:ATSproof_pferasure}
Assume that $\Sigma;\vB;\Delta\tpjg e:T$ is derivable in
$\ATSproof$. Then
$\pferase{\Sigma};\pferase{\vB},\pferasep{\Delta};\pferaset{\Delta}\tpjg\pferase{e}:\pferase{T}$
is derivable in $\ATSzero$,
\end{theorem}
\begin%
{proof}
By structural induction on the typing derivation $\Der$ of $\Sigma;\vB;\Delta\tpjg e:T$.
\begin%
{itemize}
\item
Assume that the last applied rule in $\Der$ is $\mbox{\bf(\tyrule-$\tand$-elim-pt)}$:
\[%
\begin%
{array}{c}
\infer[\mbox{\bf(\tyrule-$\tand$-elim-pt)}]
      {\Sigma;\vB;\Delta\tpjg \dletin{\dtuplept{x_1,x_2}=\ept}{e_0}:T_0}
      {\Der_1::\Sigma;\vB;\Delta\tpjg \ept: P_1 * T_2 & \Der_2::\Sigma;\vB;\Delta,x_1:P_1,x_2:T_2 \tpjg e_0:T_0}
\end{array}\]
where $e$ is $\dletin{\dtuplept{x_1,x_2}=\ept}{e_0}$ and $T=T_0$.
By induction hypothesis on $\Der_1$, there exists the following derivation in $\ATSzero$:
$$\Der'_1::\pferase{\Sigma};\pferase{\vB},\pferasep{\Delta};\pferaset{\Delta}\tpjg\pferase{\ept}:\pferase{P_1}\Band\pferase{T_2}$$
By induction hypothesis on $\Der_2$, there exists the following derivation in $\ATSzero$:
$$\Der'_2::\pferase{\Sigma};\pferase{\vB},\pferasep{\Delta},\pferase{P_1};\pferaset{\Delta},x_2:\pferase{T_2}\tpjg\pferase{e_0}:\pferase{T_0}$$
Applying the rule $\mbox{\bf(\tyrule-$\Band$-elim)}$ to $\Der'_1$ and $\Der'_2$ yields the following derivation:
$$\Der'::\pferase{\Sigma};\pferase{\vB},\pferasep{\Delta};\pferaset{\Delta}\tpjg\dsletin{\dassert{x_2}=\pferase{\ept}}{\pferase{e_0}}:\pferase{T_0}$$ Note that $\pferase{e}$
equals $\dsletin{\dassert{x_2}=\pferase{\ept}}{\pferase{e_0}}$,
and the case concludes.
\item
Assume that the last applied rule in $\Der$ is $\mbox{\bf(\tyrule-app-tp)}$:
\[%
\begin%
{array}{c}
\infer[\mbox{\bf(\tyrule-app-tp)}]
      {\Sigma; \vB; \Delta\tpjg \dapptp{e_1}{e_2}: T_2}
      {\Der_1::\Sigma; \vB; \Delta\tpjg e_1: P_1\timp T_2 & \Der_2::\Sigma; \vB; \Delta\tpjg e_2: P_1}
\end{array}\]
where $e$ is $\dapptp{e_1}{e_2}$ and $T=T_2$.
By induction hypothesis on $\Der_1$, there exists the following derivation in $\ATSzero$:
$$\Der'_1::\pferase{\Sigma};\pferase{\vB},\pferasep{\Delta};\pferaset{\Delta}\tpjg\pferase{e_1}:\pferase{P_1}\Bimp\pferase{T_2}$$
Applying Lemma~\ref{lemma:ATSproof_prop2bool} to $\Der_2$ yields that
the constraint
$\pferase{\Sigma};\pferase{\vB},\pferasep{\Delta};\pferaset{\Delta}\tpjg\pferase{P_1}$ is
valid.  Applying the rule $\mbox{\bf(\tyrule-$\Bimp$-elim)}$ to
$\Der'_1$ and the valid constraint yields the following derivation:
$$\pferase{\Sigma};\pferase{\vB},\pferasep{\Delta};\pferaset{\Delta}\tpjg\deguard{\pferase{e_1}}:\pferase{T_2}$$
Note that $\pferase{e}$ equals $\deguard{\pferase{e_1}}$, and the case
concludes.
\end{itemize}
The rest of the cases can be handled similarly.
\hfill
\end{proof}

By Theorem~\ref{theorem:ATSproof_pferasure}, the proof-erasure
of a program is well-typed in $\ATSzero$ if the program itself is
well-typed in $\ATSproof$. In other words,
Theorem~\ref{theorem:ATSproof_pferasure} justifies PwTP in $\ATSproof$
as an approach to internalizing constraint-solving through explicit
proof-construction.

\section{Related Work and Conclusion}\label{section:related}
Constructive type theory, which was originally proposed by
Martin-L{\"o}f for the purpose of establishing a foundation for
mathematics, requires pure reasoning on programs. Generalizing as well
as extending Martin-L{\"o}f's work, the framework Pure Type System
($\PTS$) offers a simple and general approach to designing and
formalizing type systems. However, type equality depends on program
equality in the presence of dependent types, making it highly
challenging to accommodate effectful programming features as these
features often greatly complicate the definition of program
equality~~\cite{CONSTABLE87,MENDLER87,HMST95,HN88}.

The framework {\em Applied Type System} ($\ATS$)~\cite{ATS-types03}
introduces a complete separation between statics, where types are
formed and reasoned about, and dynamics, where programs are
constructed and evaluated, thus eliminating by design the need for
pure reasoning on programs in the presence of dependent types. The
development of $\ATS$ primarily unifies and also extends the previous
studies on both Dependent ML (DML)~\cite{XP99,DML-jfp07} and guarded
recursive datatypes~\cite{GRDT-popl03}. DML enriches ML with a
restricted form of dependent datatypes, allowing for specification and
inference of significantly more precise type information (when
compared to ML), and guarded recursive datatypes can be thought of as
an impredicative form of dependent types in which type indexes are
themselves types. Given the similarity between these two forms of
types, it is only natural to seek a unified presentation for them.
Indeed, both DML-style dependent types and guarded recursive datatypes
are accommodated in $\ATS$.

In terms of theorem-proving, there is a fundamental difference between
$\ATS$ and various theorem-proving systems such as
NuPrl~\cite{CONSTABLE86} (based on Martin-L{\"o}f's constructive type
theory) and Coq~\cite{Dowek93tr} (based on the calculus of
construction~\cite{COQUAND88A}). In $\ATS$, proof construction is
solely meant for constraint simplification and proofs are not expected
to contain any computational meaning.  On the other hand, proofs in
NuPrl and Coq are required to be constructive as they are meant for
supporting program extraction.

The theme of combining programming with theorem-proving is also
present in the programming language
$\Omega$emga~\cite{LanguagesOfTheFuture}.  The type system of
$\Omega$emga is largely built on top of a notion called {\em equality
  constrained types} (a.k.a. phantom types~\cite{PhantomTypes}), which
are closely related to the notion of guarded recursive
datatypes~\cite{GRDT-popl03}. In $\Omega$emga, there seems no strict
separation between programs and proofs. In particular, proofs need to
be constructed at run-time. In addition, an approach to simulating
dependent types through the use of type classes in Haskell is given
in~\cite{FakingIt}, which is casually related to proof construction in
the design of $\ATS$. Please also see~\cite{CUTELIMpadl04} for a
critique on the practicality of simulating dependent types in Haskell.

In summary, a framework $\ATS$ is presented in this paper to
facilitate the design and formalization of type systems to support
practical programming.  With a complete separation between statics and
dynamics, $\ATS$ removes by design the need for pure reasoning on
programs in the presence of dependent types. Additionally, $\ATS$
allows programming and theorem-proving to be combined in a
syntactically intertwined manner, providing the programmer with an
approach to internalizing constraint-solving through explicit
proof-construction. As a minimalist formulation of $\ATS$,
$\ATSzero$ is first presented and its type-soundness formally
established. Subsequently, $\ATSzero$ is extended to $\ATSproof$ so as
to support programming with theorem-proving, and the correctness of
this extension is proven based on a translation often referred to as
proof-erasure, which turns each well-typed program in $\ATSproof$ into
a corresponding well-typed program in $\ATSzero$ of the same dynamic
semantics.

\bibliographystyle{jfp}\bibliography{mybib}

\end{document}